\pdfoutput=1
\documentclass[runningheads]{llncs}

\usepackage{amsmath,amssymb,xcolor,graphicx,wrapfig,paralist}
\usepackage{algorithmicx,algorithm}
\usepackage[noend]{algpseudocode}
\usepackage{booktabs}
\usepackage{proof}
\usepackage{tikz}
\usepackage{csquotes}
\usepackage{graphics}
\usepackage{stmaryrd}
\usepackage{soul}
\usepackage{array}
\usepackage{sectsty}

\usepackage{ellipsis, mparhack, ragged2e} % bugfixing
\usepackage[l2tabu, orthodox]{nag} % avoid typical LaTeX errors

\usepackage{tabu,multirow}
\usepackage{xparse}
\usepackage{mathtools}
\usepackage{environ}
\usepackage{lscape}

\usepackage{url}

\usepackage[english]{babel}
 \usepackage{hyperref}

\usepackage[caption=false]{subfig}

\usepackage{algorithm}
\usepackage{algpseudocode}

\usepackage{marginfix}
\usepackage{xifthen}

\usepackage[protrusion=true,expansion=true]{microtype}

\usepackage{listings}
\lstdefinestyle{customJ}{
  belowcaptionskip=1\baselineskip,
  breaklines=true,
  frame=L,
  xleftmargin=\parindent,
  language=Java,
  showstringspaces=false,
  basicstyle=\footnotesize\ttfamily,
  keywordstyle=\bfseries\color{green!40!black},
  commentstyle=\itshape\color{purple!40!black},
  identifierstyle=\color{blue},
  stringstyle=\color{orange},
}

\usepackage{afterpage}

\usepackage{xfrac}

\usepackage{etoolbox}

\usetikzlibrary{shapes,positioning,arrows,calc,automata,matrix,fit}
\tikzstyle{state}+=[minimum size = 6mm, inner sep=0,outer sep=1]
\tikzset{->,>=stealth'}
\newcommand{\highlight}[1]{\colorbox{black!15}{$\displaystyle#1$}}

\colorlet{pink}{red!40}
\colorlet{blue}{cyan!60}

\newcolumntype{L}{X}
\newcolumntype{R}{>{\raggedleft\arraybackslash}X}
\newcolumntype{C}{>{\centering\arraybackslash}X}

\newcolumntype{H}{>{\setbox0=\hbox\bgroup}c<{\egroup}@{}}

\DeclareGraphicsExtensions{.pdf, .png, .jpg}

\def\qedtriangle{\hspace{\stretch1}\ensuremath\triangleleft}
\def\qedsquare{\hspace{\stretch1}\ensuremath\square}

\newcommand{\myspace}{\vspace*{-0.5em}}

\RequirePackage{marvosym}

\DeclarePairedDelimiter{\delimabs}{\lvert}{\rvert}
\DeclarePairedDelimiter{\delimnorm}{\lVert}{\rVert}
\DeclarePairedDelimiter{\delimpospart}{\lgroup}{\rgroup^+}

\DeclarePairedDelimiterX{\deliminner}[2]{\lange}{\rangle}{#1, #2}
\DeclarePairedDelimiter{\delimcardinality}{\lvert}{\rvert}
\DeclarePairedDelimiter{\delimset}{\lbrace}{\rbrace}
\DeclarePairedDelimiter{\delimtuple}{(}{)}
\DeclarePairedDelimiter{\delimlistt}{[}{]}
\DeclarePairedDelimiter{\delimfun}{(}{)}

\NewDocumentCommand{\abs}{sm}{\IfBooleanTF{#1}{\delimabs{#2}}{\delimabs*{#2}}}
\NewDocumentCommand{\norm}{sm}{\IfBooleanTF{#1}{\delimnorm{#2}}{\delimnorm*{#2}}}
\NewDocumentCommand{\pospart}{sm}{\IfBooleanTF{#1}{\delimpospart{#2}}{\delimpospart*{#2}}}
\NewDocumentCommand{\negpart}{sm}{\IfBooleanTF{#1}{\delimnetpart{#2}}{\delimnetpart*{#2}}}
\NewDocumentCommand{\inner}{sm}{\IfBooleanTF{#1}{\deliminner{#2}}{\deliminner*{#2}}}
\NewDocumentCommand{\cardinality}{sm}{\IfBooleanTF{#1}{\delimcardinality{#2}}{\delimcardinality*{#2}}}
\NewDocumentCommand{\set}{sm}{\IfBooleanTF{#1}{\delimset*{#2}}{\delimset{#2}}}
\NewDocumentCommand{\tuple}{sm}{\IfBooleanTF{#1}{\delimtuple{#2}}{\delimtuple*{#2}}}
\NewDocumentCommand{\closure}{sm}{\IfBooleanTF{#1}{\delimclosure{#2}}{\delimclosure*{#2}}}
\NewDocumentCommand{\listt}{sm}{\IfBooleanTF{#1}{\delimlistt{#2}}{\delimlistt*{#2}}}
\NewDocumentCommand{\fun}{smm}{\IfBooleanTF{#1}{{#2}\delimfun{#3}}{{#2}\delimfun*{#3}}}
\NewDocumentCommand{\funMacro}{smm}{\IfNoValueTF{#3}{#1}{\fun{#2}{#3}}}

\DeclareMathOperator{\ExistsOp}{\exists}
\DeclareMathOperator{\ForallOp}{\forall}

\NewDocumentCommand{\Exists}{gg}{\IfNoValueTF{#1}{\ExistsOp}{\ExistsOp #1. \, #2}}
\NewDocumentCommand{\Forall}{gg}{\IfNoValueTF{#1}{\ForallOp}{\ForallOp #1. \, #2}}

\newcommand{\intersectionSym}{\cap}
\newcommand{\intersectionBin}{\mathbin{\intersectionSym}}
\newcommand{\UnionSym}{\bigcup}

\newcommand{\intersection}{\intersectionBin}
\newcommand{\Union}{\UnionSym}

\newcommand{\Naturals}{\mathbb{N}}

\newcommand{\Reals}{\mathbb{R}}

\newcommand{\Distributions}{\mathcal{D}}

\NewDocumentCommand{\convto}{G{}}{\xrightarrow{#1}}
\NewDocumentCommand{\weakto}{G{}}{\xrightharpoonup{#1}}
\NewDocumentCommand{\weakstarto}{G{}}{\xrightharpoonup[*]{#1}}

 \DeclareDocumentCommand{\diff}{D<>{} O{}  D(){}}{\Delta_{#1}^{#2}\ifthenelse{\isempty{#3}}{}{(#3)}}

\DeclareDocumentCommand{\post}{D<>{} O{} D(){}}{\mathsf{Post}_{#1}^{#2}\ifthenelse{\isempty{#3}}{}{(#3)}}
\DeclareMathOperator{\leaves}{\mathbin{\mathop{exits}}}

\newcommand{\eqdef}{\vcentcolon=}

\NewDocumentCommand{\distributions}{d()}{\funMacro{\mathcal{D}}{#1}}

\newcommand{\pmin}{p_{\min}}

\DeclareDocumentCommand{\val}{D<>{} O{}  D(){} t'}{\mathsf{V}_{#1}^{\IfBooleanTF{#4}{\prime #2}{#2}}\ifthenelse{\isempty{#3}}{}{(#3)}}
\DeclareDocumentCommand{\ub}{D<>{} O{}  D(){} t'}{\mathsf{U}_{#1}^{\IfBooleanTF{#4}{\prime #2}{#2}}\ifthenelse{\isempty{#3}}{}{(#3)}}
\DeclareDocumentCommand{\gub}{D<>{} O{}  D(){} t'}{\mathsf{G}_{#1}^{\IfBooleanTF{#4}{\prime #2}{#2}}\ifthenelse{\isempty{#3}}{}{(#3)}}
\DeclareDocumentCommand{\lb}{D<>{} O{}  D(){} t'}{\mathsf{L}_{#1}^{\IfBooleanTF{#4}{\prime #2}{#2}}\ifthenelse{\isempty{#3}}{}{(#3)}}
\DeclareDocumentCommand{\game}{D<>{} O{} D(){} t'}{\mathsf{G}_{#1}^{\IfBooleanTF{#4}{\prime}{}#2}\ifthenelse{\isempty{#3}}{}{(#3)}}
\DeclareDocumentCommand{\transition}{D<>{} O{} D(){}}{\rightarrow_{#1}^{#2}\ifthenelse{\isempty{#3}}{}{(#3)}}

\newcommand{\SG}{\textrm{SG}}

\DeclareDocumentCommand{\G}{D<>{} O{} t' D(){}}{\mathsf{G}_{#1}^{\IfBooleanTF{#3}{\prime}{}#2}\ifthenelse{\isempty{#4}}{}{(#4)}}
\DeclareDocumentCommand{\exGame}{D<>{} O{} t'  D(){}}{\mathsf{G}_{#1}^{\IfBooleanTF{#3}{\prime}{}#2}\ifthenelse{\isempty{#4}}{}{(#4)}=(\states<#1>[\IfBooleanTF{#3}{\prime}{}#2],\states<\Box\ifthenelse{\isempty{#1}}{}{,#1}>[\IfBooleanTF{#3}{\prime}{}#2],\states<\circ\ifthenelse{\isempty{#1}}{}{,#1}>[\IfBooleanTF{#3}{\prime}{}#2],\istate<#1>[\IfBooleanTF{#3}{\prime}{}#2],\actions<#1>[\IfBooleanTF{#3}{\prime}{}#2],\Av<#1>[\IfBooleanTF{#3}{\prime}{}#2],\trans<#1>[\IfBooleanTF{#3}{\prime}{}#2])}

\DeclareDocumentCommand{\states}{D<>{} O{}  t'}{\mathsf{S}_{#1}^{\IfBooleanTF{#3}{\prime~#2}{#2}}}
\DeclareDocumentCommand{\state}{D<>{} O{}  t'}{\mathsf{s}_{#1}^{\IfBooleanTF{#3}{\prime #2}{#2}}}
\DeclareDocumentCommand{\istate}{D<>{} O{} t'}{\mathsf{s}_{0\ifthenelse{\isempty{#1}}{}{,#1}}^{\IfBooleanTF{#3}{\prime~#2}{#2}}}

\newcommand{\initstate}{\state<0>}

\DeclareDocumentCommand{\trans}{D<>{} O{} t' D(){} D(){}}{\mathbb{T}{#1}^{\IfBooleanTF{#3}{\prime}{}#2}\ifthenelse{\isempty{#4}}{}{(#4)}\ifthenelse{\isempty{#5}}{}{(#5)}}
\DeclareDocumentCommand{\Av}{D<>{} O{} t' D(){}}{\mathsf{Av}_{#1}^{\IfBooleanTF{#3}{\prime}{}#2}\ifthenelse{\isempty{#4}}{}{(#4)}}

\DeclareDocumentCommand{\F}{D<>{} O{} t' D(){}}{\mathsf{F}_{#1}^{\IfBooleanTF{#3}{\prime}{}#2}\ifthenelse{\isempty{#4}}{}{(#4)}}

\newcommand{\Path}{\rho}
\DeclareDocumentCommand{\Path}{D<>{} O{} t' D(){}}{\path<#1>[#2]\IfBooleanTF{#3}{'}{}(#4)}
\DeclareDocumentCommand{\path}{D<>{} O{} t' D(){}}{\rho_{#1}^{\IfBooleanTF{#3}{\prime}{}#2}\ifthenelse{\isempty{#4}}{}{(#4)}}
\newcommand{\fpath}{\mathsf{w}}
\DeclareDocumentCommand{\Paths}{D<>{} O{} t' D(){}}{\Omega_{#1}^{\IfBooleanTF{#3}{\prime}{}#2}\ifthenelse{\isempty{#4}}{}{(#4)}}
\newcommand{\straa}{\sigma}

\newcommand{\strab}{\tau}

\DeclareDocumentCommand{\strategy}{D<>{} O{} D(){}
  t*}{{\IfBooleanTF{#4}{\tau}{\sigma}}_{#1}^{#2}\ifthenelse{\isempty{#3}}{}{(#3)}}

\DeclareDocumentCommand{\actions}{D<>{} O{} t' d()}{{\IfNoValueTF{#4}{\mathsf{A}}{\fun{\mathsf{A}}{#4}}}_{#1}^{\IfBooleanTF{#3}{\prime~#2}{#2}}}
\DeclareDocumentCommand{\action}{D<>{} O{} t'}{\mathsf{a}_{#1}^{\IfBooleanTF{#3}{\prime#2}{#2}}}

\newcommand{\mec}{\mathsf{MEC}}

\newcommand{\pr}{\mathbb P}
\newcommand{\drawcirc}{\node[draw,circle,minimum size=.7cm, outer sep=1pt]}
\newcommand{\drawbox}{\node[draw,rectangle,minimum size=.7cm, outer sep=1pt]}
\newcommand{\drawdummy}{\node[minimum size=0,inner sep=0]}

\newcommand{\para}[1]{\noindent\textbf{#1} }

\newcommand{\sinks}{\mathsf{Zero}}

\newcommand{\best}{\mathsf{best}}
\newcommand{\Vis}{\widehat{\states}}
\DeclareDocumentCommand{\target}{D<>{} O{}
  t'}{\mathfrak{1}_{#1}^{\IfBooleanTF{#3}{\prime}{}#2}}
\NewDocumentCommand{\exit}{D<>{}D[]{}D(){}}{\mathsf{bestExit}_{#1}^{#2}\ifthenelse{\isempty{#3}}{}{(#3)}}

\newcommand{\postmax}{\post<\max>}
\newcommand{\targetset}{\mathsf{Goal}}
\newcommand{\sink}{\mathfrak{0}} %set of states with no path to target, needed for bellman equations

\newcommand{\Prob}{\mathbb{P}}

\newcommand{\tsucc}{\mathsf{t}}

\newcommand{\INITIALIZE}{\mathsf{INITIALIZE\_BOUNDS}}

\newcommand{\UPDATE}{\mathsf{UPDATE}}
\newcommand{\SIMULATE}{\mathsf{SIMULATE}}
\newcommand{\STUCK}{\mathsf{LOOPING}}
\newcommand{\SUREEC}{\mathsf{\deltasure~EC}}
\newcommand{\FIND}{\mathsf{FIND\_MSECs}}
\newcommand{\DEFLATE}{\mathsf{DEFLATE}}

\newcommand{\transdelta}{\delta_{\trans}}
\newcommand{\deltasure}{\transdelta\textit{-sure}}
\newcommand{\stepsUntilSure}{\textit{requiredSamples}}
\newcommand{\theN}{\mathcal{N}_k}
\newcommand{\deltaIter}{\delta_{k}}

\newtoggle{arxiv}

\newcommand{\arxivcite}[2]{\iftoggle{arxiv}{#1}{#2}}

\toggletrue{arxiv}

\advance\textwidth3mm % <= 5mm
\advance\hoffset-1.5mm
\advance\textheight1mm
\makeatletter
\renewcommand\section{\@startsection {section}{1}{\z@}%
      {-2ex \@plus -1ex \@minus -.2ex}% <beforeskip>
      {1ex \@plus .2ex}% <afterskip>
      {\normalfont\Large\bfseries\SS@sectfont}}
\renewcommand\subsection{\@startsection{subsection}{2}{\z@}%
      {-2ex\@plus -1ex \@minus -.2ex}% <beforeskip>
      {1ex \@plus .2ex}% <afterskip>
      {\normalfont\large\bfseries\SS@subsectfont}}
\makeatother

\renewcommand{\circ}{\bigcirc}

\usepackage[firstpage]{draftwatermark}
\SetWatermarkText{\hspace*{5in}\raisebox{6in}{\includegraphics[scale=0.1]{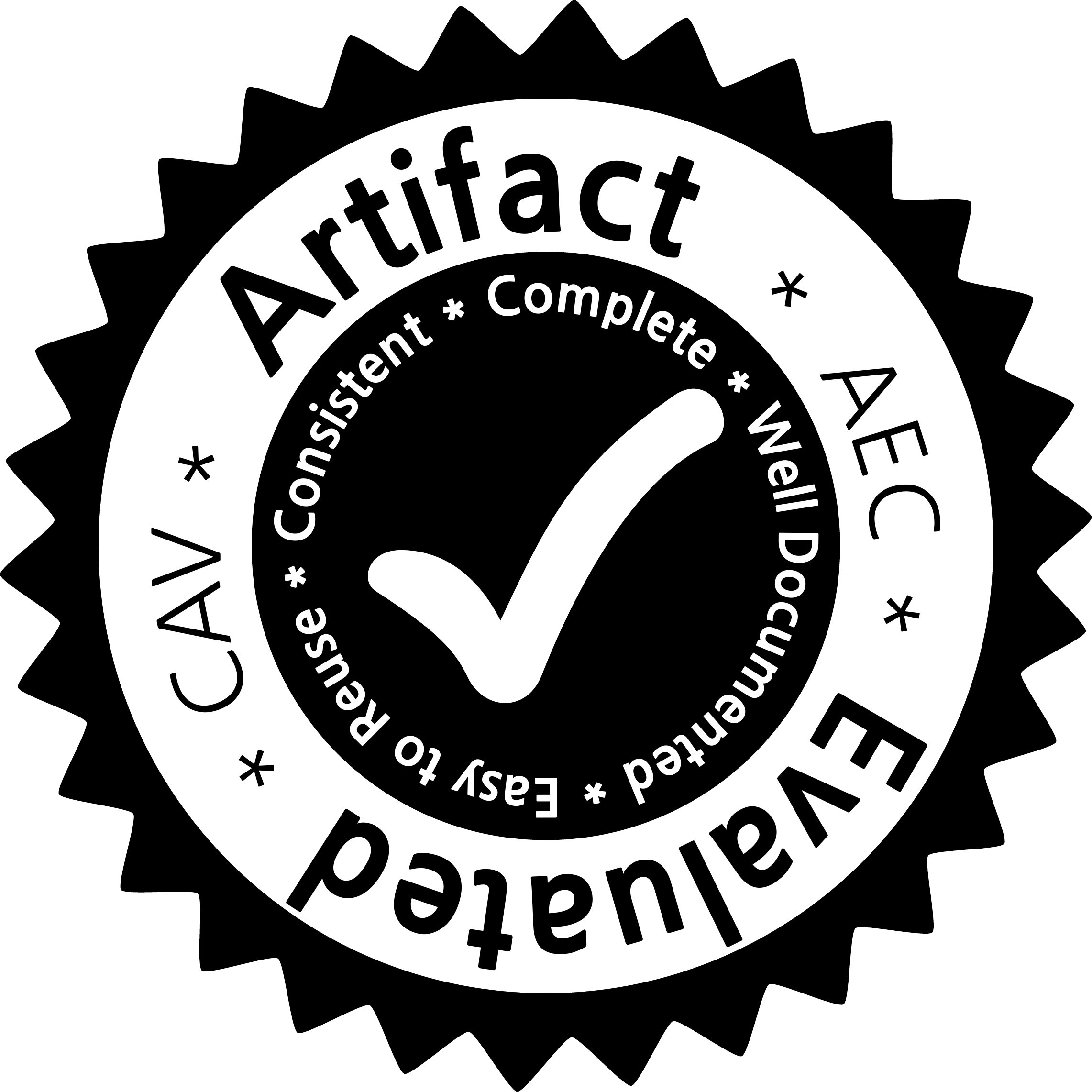}}}
\SetWatermarkAngle{0}

\title{PAC Statistical Model Checking for Markov Decision Processes and Stochastic Games%
\thanks{This research was funded in part by TUM IGSSE Grant 10.06 (PARSEC), the Czech Science Foundation grant No. 18-11193S, and the German Research Foundation (DFG) project KR 4890/2-1 ``Statistical Unbounded Verification''. We thank Florent Delgrange for his valuable feedback on the proof of Theorem 2.}
}
\author{Pranav Ashok \and Jan K\v ret\'insk\'y \and Maximilian Weininger}
\institute{Technical University of Munich, Germany} 

\begin{document}

\maketitle

\myspace
\myspace
\myspace
\myspace

\begin{abstract}
Statistical model checking (SMC) is a technique for analysis of probabilistic systems that may be (partially) unknown.
We present an SMC algorithm for (unbounded) reachability yielding probably approximately correct (PAC) guarantees on the results.
We consider both the setting (i) with no knowledge of the transition function (with the only quantity required a bound on the minimum transition probability) and (ii) with knowledge of the topology of the underlying graph.
On the one hand, it is the first algorithm for stochastic games.
On the other hand, it is the first practical algorithm even for Markov decision processes.
Compared to previous approaches where PAC guarantees require running times longer than the age of universe even for systems with a handful of states, our algorithm often yields reasonably precise results within minutes, not requiring the knowledge of mixing time.
\end{abstract} 

\myspace
\myspace

\section{Introduction}

\para{Statistical model checking (SMC)} \cite{YS02} is an analysis technique for probabilistic systems based on 
\begin{compactenum}
	\item simulating finitely many finitely long runs of the system,
	\item statistical analysis of the obtained results,
	\item yielding a confidence interval/probably approximately correct (PAC) result on the probability of satisfying a given property, i.e., there is a non-zero probability that the bounds are incorrect, but they are correct with probability that can be set arbitrarily close to $1$.
\end{compactenum}
One of the advantages is that it can avoid the state-space explosion problem, albeit at the cost of weaker guarantees.
Even more importantly, this technique is applicable even when the model is not known (\emph{black-box} setting) or only qualitatively known (\emph{grey-box} setting), where the exact transition probabilities are unknown such as in many cyber-physical systems.

In the basic setting of Markov chains \cite{norris1998markov} with (time- or step-)bounded properties, the technique is very efficient and has been applied to numerous domains, e.g.\ biological \cite{DBLP:conf/cmsb/JhaCLLPZ09,DBLP:conf/cmsb/PalaniappanG0HT13}, hybrid \cite{DBLP:conf/hybrid/ZulianiPC10,DBLP:journals/corr/abs-1208-3856,DBLP:conf/formats/EllenGF12,DBLP:conf/formats/Larsen12} or cyber-physical \cite{DBLP:conf/forte/BasuBBCDL10,DBLP:conf/atva/ClarkeZ11,DBLP:conf/nfm/DavidDLLM13} systems and a substantial tool support is available \cite{DBLP:conf/tacas/JegourelLS12,DBLP:journals/corr/abs-1207-1272,DBLP:conf/qest/BoyerCLS13,DBLP:conf/mmb/BogdollHH12}.
In contrast, whenever either (i)~infinite time-horizon properties, e.g. reachability, are considered or (ii)~non-determinism is present in the system, providing any guarantees becomes significantly harder.

Firstly, for \emph{infinite time-horizon properties} we need a stopping criterion such that the infinite-horizon property can be reliably evaluated based on a finite prefix of the run yielded by simulation.
This can rely on the the complete knowledge of the system (\emph{white-box} setting) \cite{sbmf11,DBLP:journals/apal/LassaigneP08}, the topology of the system (grey box) \cite{sbmf11,ase10}, or a lower bound $\pmin$ on the minimum transition probability in the system (black box) \cite{DHKP16,BCC+14}.

Secondly, for Markov decision processes (MDP) \cite{Puterman} with (non-trivial) \emph{non-determinism}, \cite{DBLP:conf/qest/HenriquesMZPC12} and \cite{LP12} employ reinforcement learning~\cite{SB98} in the setting of bounded properties or discounted (and for the purposes of approximation thus also bounded) properties, respectively.
The latter also yields PAC guarantees.

Finally, for MDP with unbounded properties, \cite{DBLP:conf/forte/BogdollFHH11} deals with MDP with spurious non-determinism, where the way it is resolved does not affect the desired property.
The general non-deterministic case is treated in \cite{DBLP:conf/rss/FuT14,BCC+14}, yielding PAC guarantees.
However, the former requires the knowledge of mixing time, which is at least as hard to compute; the algorithm in the latter is purely theoretical since before a single value is updated in the learning process, one has to simulate longer than the age of universe even for a system as simple as a Markov chain with 12 states having at least 4 successors for some state.

\para{Our contribution} is an SMC algorithm with PAC guarantees for (i) MDP and unbounded properties, which runs for realistic benchmarks \cite{qcomp} and confidence intervals in orders of minutes, and (ii) is the first algorithm for stochastic games (SG).
It relies on different techniques from literature.
\begin{enumerate}
	\item The increased practical performance rests on two pillars:
	\begin{itemize}
		\item extending early detection of bottom strongly connected components in Markov chains by \cite{DHKP16} to end components for MDP and simple end components for SG;
		\item improving the underlying PAC Q-learning technique of \cite{Strehl}: 
		\begin{enumerate}
			\item learning is now model-based with better information reuse instead of model-free, but in realistic settings with the same memory requirements,
			\item better guidance of learning due to interleaving with precise computation, which yields more precise value estimates.
			\item splitting confidence over all relevant transitions, allowing for variable width of confidence intervals on the learnt transition probabilities.
		\end{enumerate}
	\end{itemize}
	\item The transition from algorithms for MDP to SG is possible via extending the over-approximating value iteration from MDP \cite{BCC+14} to SG by \cite{KKKW18}.
\end{enumerate}
To summarize, we give an anytime PAC SMC algorithm for (unbounded) reachability.
It is the first such algorithm for SG and the first practical one for MDP.

%For this, we had to do several things
%\begin{itemize}
%    \item Improve UPDATE from \cite{BCC+14} to be feasible
%    \item Extend (and improve?) BSCC detection from \cite{DHKP16} to detect SECs quickly
%    \item Use the methods of \cite{KKKW18} to handle ECs/SECs.
%    \item Use the two-phase approach to get statistical guarantees.
%\end{itemize}

\subsection*{Related work}
%The stuff here \url{https://docs.google.com/spreadsheets/d/1IhXLPRPCWdl2M_ZvDePZB1F8nQXtxC6RLEx7H03902E/edit#gid=1279994026}, and more.

Most of the previous efforts in SMC have focused on the analysis of properties 
with \emph{bounded} horizon 
\cite{Younes02,Sen04,DBLP:journals/sttt/YounesKNP06,DBLP:conf/cmsb/JhaCLLPZ09,DBLP:conf/tacas/JegourelLS12,DBLP:journals/corr/abs-1207-1272}.

SMC of \emph{unbounded} properties was first considered in \cite{vmcai04} and the first approach was proposed in \cite{cav05}, but observed incorrect in \cite{ase10}.
Notably, in~\cite{sbmf11} two approaches are described. 
The {first approach} proposes to terminate sampled paths at every step with some probability $p_{term}$ and re-weight the result accordingly.
In order to guarantee the asymptotic convergence of this method, the second eigenvalue $\lambda$ of 
the chain and its mixing time must be computed, which is as hard as the verification problem itself and requires the complete knowledge of the system ({white box} setting).
The correctness of \cite{DBLP:journals/apal/LassaigneP08} 
relies on the knowledge of the second eigenvalue $\lambda$, too.
The {second approach} of~\cite{sbmf11} requires the knowledge of the chain's 
topology (grey box), which is used to transform the chain so that all potentially 
infinite paths are eliminated.
In \cite{ase10}, a similar transformation is performed, again requiring 
knowledge of the topology.
In \cite{DHKP16}, only (a~lower bound on) the minimum transition probability $\pmin$ is assumed and PAC guarantees are derived.
While unbounded properties cannot be analyzed without any  information on the system, knowledge of $\pmin$ is a relatively light assumption in many realistic scenarios \cite{DHKP16}. 
For instance, bounds on the rates for reaction kinetics in chemical reaction systems are typically known; for models in the PRISM language \cite{prism}, the bounds can be easily inferred without constructing the respective state space.
In this paper, we thus adopt this assumption.

	In the case with general \emph{non-determinism},
	one approach is to give the non-determinism a probabilistic semantics,
	e.g., using a uniform distribution instead, as for timed automata
	in \cite{DBLP:conf/formats/DavidLLMPVW11,DBLP:conf/cav/DavidLLMW11,DBLP:conf/ifm/Larsen13}.
	Others~\cite{LP12,DBLP:conf/qest/HenriquesMZPC12,BCC+14} aim to quantify over all strategies and produce an $\epsilon$-optimal strategy.
	In \cite{DBLP:conf/qest/HenriquesMZPC12}, candidates for optimal strategies are generated and gradually improved, but ``at any given point we cannot quantify how close to optimal the candidate scheduler is'' (cited from~\cite{DBLP:conf/qest/HenriquesMZPC12}) and the algorithm 
	%does not estimate the maximum probability of the property'' 
	``does not in general converge to the true optimum''
	(cited from~\cite{DBLP:conf/sefm/LegayST14}). Further, \cite{DBLP:conf/sefm/LegayST14,DBLP:journals/sttt/DArgenioLST15,DBLP:conf/isola/DArgenioHS18} randomly sample compact representation of strategies, resulting in useful lower bounds if $\varepsilon$-schedulers are frequent.
	\cite{DBLP:conf/tacas/HahnPSSTW19} gives a convergent model-free algorithm (with no bounds on the current error) and identifies that the previous \cite{DBLP:conf/cdc/SadighKCSS14} ``has two faults, the second of which also affects approaches [...] \cite{DBLP:journals/corr/abs-1801-08099,DBLP:journals/corr/abs-1902-00778}''.

	Several approaches provide SMC for MDPs and unbounded properties with \emph{PAC guarantees}.
Firstly, similarly to \cite{DBLP:journals/apal/LassaigneP08,sbmf11}, \cite{DBLP:conf/rss/FuT14} requires (1) the mixing time $T$ of the MDP. The algorithm then yields PAC bounds in time polynomial in $T$ (which in turn can of course be exponential in the size of the MDP). Moreover, the algorithm requires (2) the ability to restart simulations also in non-initial states, (3) it only returns the strategy once all states have been visited (sufficiently many times), and thus (4) requires the size of the state space $|S|$.
Secondly, \cite{BCC+14}, based on delayed Q-learning (DQL) \cite{Strehl}, lifts the assumptions (2) and (3) and instead of (1) mixing time requires only (a bound on) the minimum transition probability $\pmin$.
Our approach additionally lifts the assumption (4) and allows for running times faster than those given by $T$, even without the knowledge of $T$.

    Reinforcement learning (without PAC bounds) for stochastic games has been considered already in \cite{DBLP:journals/mor/LakshmivarahanN81,DBLP:conf/icml/Littman94,DBLP:conf/ijcai/BrafmanT99}.
	\cite{DBLP:conf/ijcai/WenT16} combines the special case of almost-sure satisfaction of a specification with optimizing quantitative objectives.
	We use techniques of \cite{KKKW18}, which however assumes access to the transition probabilities.

\section{Preliminaries}

\subsection{Stochastic games}
A \emph{probability distribution} on a finite set $X$ is a mapping $\delta: X \to [0,1]$, such that $\sum_{x\in X} \delta(x) = 1$.
The set of all probability distributions on $X$ is denoted by $\Distributions(X)$.
Now we define turn-based two-player stochastic games.
As opposed to the notation of e.g.\ \cite{condonComplexity}, we do not have special stochastic nodes, but rather a probabilistic transition function.
\begin{definition}[\SG]
	A \emph{stochastic game ($\SG$)} is a tuple 
	$\exGame$,
	where $\states$ is a finite set of \emph{states} partitioned\footnote{I.e., $\states<\Box> \subseteq \states$, $\states<\circ> \subseteq \states$, $\states<\Box> \cup \states<\circ> = \states$, and $\states<\Box> \cap \states<\circ> = \emptyset$.} 
	into the sets $\states<\Box>$ and $\states<\circ>$ of states of the player \emph{Maximizer} and \emph{Minimizer}\footnote{The names are chosen, because Maximizer maximizes the probability of reaching a given target state, and Minimizer minimizes it.}, respectively
	$\initstate \in \states$ is the \emph{initial} state, $\actions$ is a finite set of \emph{actions}, $\Av: \states \to 2^{\actions}$ assigns to every state a set of \emph{available} actions, and $\trans: \states \times \actions \to \Distributions(\states)$ 
	is a \emph{transition function} that given a state $\state$ and an action $\action\in \Av(\state)$  yields a probability distribution over \emph{successor} states.
	Note that for ease of notation we write $\trans(\mathsf s,\mathsf a,\mathsf t)$ instead of $\trans(\mathsf s,\mathsf a)(\mathsf t)$. 
\end{definition}
A {Markov decision process (MDP)} is a special case of $\SG$ where $\states<\circ> = \emptyset$.
A Markov chain (MC) can be seen as a special case of an MDP, where for all $\state \in \states: \abs{\Av(\state)} = 1$.
We assume that $\SG$ are non-blocking, so for all states $\state$ we have $\Av(\state) \neq \emptyset$.
%Note that, since we will operate in black, we cannot assume any pre-processing that unifies targets and sinks

For a state $\state$ and an available action $\action \in \Av(\state)$, we denote the set of successors by $\post(\state,\action) := \set{\tsucc \mid \trans(\state,\action,\tsucc) > 0}$. 
We say a state-action pair $(\state, \action)$ is an \emph{exit} of a set of states $T$, written $(\state,\action)\leaves T$, if $\exists \tsucc \in \post(\state,\action):\tsucc \notin T$, i.e., if with some probability a successor outside of $T$ could be chosen.
%For any set of states $T \subseteq \states$, we use $T_\Box$ and $T_\circ$ to denote the states of $T$ that belong to Maximizer and Minimizer, whose states are drawn in the figures as $\Box$ and $\circ$, respectively. => Only needed for bestExit, so only state there

We consider algorithms that have a limited information about the SG.

\begin{definition}[Black box and grey box]\label{def:limit}
An algorithm inputs an SG as \emph{black box} if it cannot access the whole tuple, but
\begin{itemize}
	\item it knows the initial state,
	\item for a given state, an oracle returns its player and available action,
	\item given a state $\state$ and action $\action$, it can sample a successor $\tsucc$ according to $\trans(\state,\action)$,\footnote{Up to this point, this definition conforms to black box systems in the sense of \cite{Sen04} with sampling from the initial state, being slightly stricter than \cite{Younes02} or \cite{atva09}, where simulations can  be run from any desired state.
	Further, we assume that we can choose actions for the adversarial player or that she plays fairly. Otherwise the adversary could avoid playing her best strategy during the SMC, not giving SMC enough information about her possible behaviours.}

    \item it knows $p_{\min} \leq \min_{\substack{\state \in \states,\action \in \Av(\state)\\ \tsucc \in \post(\state,\action)}} \trans(\state,\action,\tsucc)$, an under-approximation of the minimum transition probability.
\end{itemize}
When input as \emph{grey box} it additionally knows the number $\abs{\post(\state,\action)}$ of successors for each state $\state$ and action $\action$.\footnote{This requirement is slightly weaker than the knowledge of the whole topology, i.e. $\post(\state,\action)$ for each $\state$ and $\action$.}
\end{definition}

%\begin{remark}\todo{polish}
%We assume the adversarial states are "playable/observable" in some sense.
%I.e. during simulation either we can really choose the action in that state (which can be sensible in settings like e.g. chess; we know what the opponent can do and hence can simulate him); or the opponent plays random in some fair manner and we see which action he uses; then, given enough simulations, we can always get more simulations for a certain action. 
%If the opponent was already adversarial during the simulations, he could hide a part of the state space from us, so we do not learn anything about it, and then when applying the synthesized controller, he goes there and we underperform for lack of information.
%\end{remark}

The semantics of SG is given in the usual way by means of strategies and the induced Markov chain \cite{BaierBook} and its respective probability space, as follows.
An \emph{infinite path} $\path$ is an infinite sequence $\path = \state<0> \action<0> \state<1> \action<1> \cdots \in (\states \times \actions)^\omega$, such that for every $i \in \Naturals$, $\action<i>\in \Av(\state<i>)$ and $\state<i+1> \in \post(\state<i>,\action<i>)$.
%\emph{Finite path}s are defined analogously as elements of $(\states \times \actions)^\ast \times \states$.

%Since this paper deals with the reachability objective, we can restrict our attention to memoryless (positional) strategies, which are optimal for this objective~\cite{gameSurvey}\todo{citation ok?}.
%We still allow randomizing strategies, because they are needed for the simulation-based algorithm.

A \emph{strategy} of Maximizer or Minimizer is a function $\straa: \states<\Box> \to \distributions(\actions)$ or $\states<\circ> \to \distributions(\actions)$, respectively, such that $\straa(\state) \in \distributions(\Av(\state))$ for all $\state$.
Note that we restrict to memoryless strategies, as they suffice for reachability in SGs~\cite{gameSurveyKrish}. 
%We call a strategy \emph{deterministic} if it maps to Dirac distributions only; otherwise, it is \emph{randomizing}.
%Although randomization is not necessary for reachability objective in stochastic games, our algorithm uses it as it allows for simpler description.
A pair $(\straa,\strab)$ of strategies of Maximizer and Minimizer induces a Markov chain $\G[\straa,\strab]$ with states $\states$, $\initstate$ being initial, and the transition function $\trans(\state)(\tsucc) = \sum_{\action \in \Av(\state)} \straa(\state) (\action) \cdot \trans(\state, \action,\tsucc)$ for states of Maximizer and analogously for states of Minimizer, with $\straa$ replaced by $\strab$.
The Markov chain %together with a state $\state$ 
induces a unique probability distribution $\pr^{\straa,\strab}$ over measurable sets of infinite paths \cite[Ch.~10]{BaierBook}. 

\subsection{Reachability objective}\label{sec:prelimReach}
For a goal set $\targetset\subseteq\states$, we write $\Diamond \targetset:=\set{\state<0> \action<0> \state<1> \action<1> \cdots  \mid \exists i \in \Naturals: \state<i>
	\in\targetset}$ to denote the (measurable) set of all infinite paths which eventually reach $\targetset$.
For each $\state\in\states$, we define
the \emph{value} in $\state$ as 
\[\val(\state) \eqdef \sup_{\straa} \inf_{\strab} \pr_{s}^{\straa,\strab}(\Diamond \targetset)= \inf_{\strab} \sup_{\straa}\pr_{s}^{\straa,\strab}(\Diamond \targetset),\]

\noindent where the equality follows from~\cite{leastReadablePaperJanEverRead}.
We are interested in $\val(\initstate)$, its $\varepsilon$-approximation and the corresponding ($\varepsilon$-)optimal strategies for both players.

Let $\sinks$ be the set of states, from which there is no finite path to any state in $\targetset$. 
The value function $\val$ satisfies the following system of equations, which is referred to as the \emph{Bellman equations}:
\begin{equation*}\label{eq:Vs}
\val(\state) =  \begin{cases} \max_{\action \in \Av(\state)}\val(\state,\action)		&\mbox{if } \state \in \states<\Box>  \\
\min_{\action \in \Av(\state)}\val(\state,\action) &\mbox{if } \state \in \states<\circ>\\
1 &\mbox{if } \state \in \targetset \\
0 &\mbox{if } \state \in \sinks
\end{cases}
\end{equation*}
with the abbreviation $\val(\state,\action) \eqdef \sum_{s' \in S} \trans(\state,\action,\state') \cdot \val(\state')$.
Moreover, $\val$ is the \emph{least} solution to the Bellman equations, see e.g. \cite{visurvey}.

\subsection{Bounded and asynchronous value iteration}\label{sec:prelimAlgo}
The well known technique of value iteration, e.g. \cite{Puterman,gameSurvey}, works by starting from an under-approximation of value function and then applying the Bellman equations. 
This converges towards the least fixpoint of the Bellman equations, i.e. the \emph{value function}.
Since it is difficult to give a convergence criterion, the approach of bounded value iteration (BVI, also called interval iteration) was developed for MDP \cite{BCC+14,HM17} and SG~\cite{KKKW18}.
Beside the under-approximation, it also updates an over-approximation according to the Bellman equations.
The most conservative over-approximation is to use an upper bound of 1 for every state.
For the under-approximation, we can set the lower bound of target states to 1; all other states have a lower bound of 0.
We use the function $\INITIALIZE$ in our algorithms to denote that the lower and upper bounds are set as just described; 
see Algorithm \ref{alg:init} in \arxivcite{Appendix \ref{app:init}}{\cite[Appendix A.1]{arxiv}} for the pseudocode.
Additionally, BVI ensures that the over-approximation converges to the least fixpoint by taking special care of \emph{end components}, which are the reason for not converging to the true value from above.

\begin{definition}[End component(EC)]
\label{def:EC}
A non-empty set $T\subseteq \states$ of states is an \emph{end component (EC)} if there is a non-empty set $B \subseteq \Union_{\state \in T} \Av(s)$ of actions such that 
(i) for each $\state \in T, \action \in B \intersection \Av(\state)$ we do \emph{not} have $(\state,\action) \leaves T$ and
(ii) for each $\state, \state' \in T$ there is a finite path $\fpath = \state \action<0> \dots \action<n> \state' \in (T \times B)^* \times T$, i.e. the path stays inside $T$ and only uses actions in $B$.
\end{definition}

Intuitively, ECs correspond to bottom strongly connected components of the Markov chains induced by possible strategies, so for some pair of strategies all possible paths starting in the EC remain there. 
An end component $T$ is a \emph{maximal end component (MEC)} if there is no other end component $T'$ such that $T \subseteq T'$.
Given an $\SG$ $\G$, the set of its MECs is denoted by $\mec(\G)$.

Note that, to stay in an EC in an SG, the two players would have to cooperate, since it depends on the pair of strategies.
To take into account the adversarial behaviour of the players, it is
also relevant to look at a subclass of ECs, the so called \emph{simple end components}, introduced in \cite{KKKW18}.
\begin{definition}[Simple end component (SEC) \cite{KKKW18}]
An EC $T$ is called \emph{simple}, if for all $\state \in T$ it holds that $\val(\state) = \exit(T,\val)$, where 
\[
\exit(T,f) :=  
\begin{cases}
1 &\mbox{if } T \cap \targetset \neq \emptyset\\
\max_{\substack{\state \in T \cap \states<\Box>\\ (\state,\action) \leaves T}} f(\state,\action) & \mbox{else}
\end{cases}
\] 
is called  the \emph{best exit} (of Maximizer) from $T$ according to the function $f: \states \to \Reals$.
To handle the case that there is no exit of Maximizer in $T$ we set $\max_\emptyset = 0$.
\end{definition}

Intuitively, SECs are ECs where Minimizer does not want to use any of her exits, as all of them have a greater value than the best exit of Maximizer.
Assigning any value between those of the best exits of Maximizer and Minimizer to all states in the EC is a solution to the Bellman equations, because both players prefer remaining and getting that value to using their exits \cite[Lemma 1]{KKKW18}. 
However, this is suboptimal for Maximizer, as the goal is not reached if the game remains in the EC forever.
Hence we ``deflate'' the upper bounds of SECs, i.e. reduce them to depend on the best exit of Maximizer.
$T$ is called maximal simple end component (MSEC), if there is no SEC $T'$ such that $T \subsetneq T'$.
Note that in MDPs, treating all MSECs amounts to treating all MECs.
%, since if the player is maximizing, he can choose to use the $\exit$, and if the player is minimizing, he can remain in every EC and achieve the value 0, which corresponds to $\exit$\textcolor{red}{, since there are no states of Maximizer and hence trivially no exits.
%The reason why SECs prevent the over-approximation from converging is, that there is a cyclic dependency between the states in it.
%Minimizer's best strategy is to remain in the SEC
%however, given an over-approximation Maximizer is under the illusion that remaining in the EC yields a higher value, since the 
%\textcolor{blue}{Intuitively, a SEC is an EC in the game with no suboptimal (according to $f$) actions of Minimizer, i.e. Minimizer's best strategy is to remain in this EC.
%Then all the values of states in the EC are equal to that of the $\exit$ of Maximizer.
%Staying in the SEC yields value 0 (except if a target was in the SEC), so the best maximizing strategy is to leave, achieving best exit.
%
%Note that in contrast to the definition of \cite{KKKW18}, the definition of SEC in this paper includes a check whether a target is in the EC, and if so, sets the value to 1. 
%This is necessary, because in \cite{KKKW18} the SG was preprocessed to only have a single target state that is a sink, while in the limited information setting the SG cannot be preprocessed. \todopranav{I think these lines referencing KKKW18 can be removed and somewhere say that we are adapting the definition of KKKW18.}
%}

\begin{algorithm}[htbp]
	\caption{Bounded value iteration algorithm for SG (and MDP)}\label{alg:general}
	\begin{algorithmic}[1]
		\Procedure{BVI}{SG $\G$, target set $\targetset$, precision $\epsilon>0$}
		\State $\INITIALIZE$
		
		\Repeat
		\State $X \gets \SIMULATE$ \textit{until} $\STUCK$ or state in $\targetset$ is hit
		    \label{line:relevantStates}
		\State $\UPDATE(X)$ ~~~~~~~~~~\Comment{Bellman updates or their modification}
		\For {$T \in \FIND(X)$}
		    \State $\DEFLATE(T)$ \Comment{Decrease the upper bound of MSECs}
		\EndFor
		\Until{$\ub(\initstate) - \lb(\initstate) < \epsilon$}
		\EndProcedure
	\end{algorithmic}
\end{algorithm}

Algorithm \ref{alg:general} rephrases that of \cite{KKKW18} and describes the general structure of all bounded value iteration algorithms that are relevant for this paper. 
We discuss it here since all our improvements refer to functions (in capitalized font) in it.
In the next section, we design new functions, pinpointing the difference to the other papers.
The pseudocode of the functions adapted from the other papers can be found, for the reader's convenience, in \arxivcite{Appendix \ref{app:algos}}{\cite[Appendix A]{arxiv}}. % contains the pseudocode for all the functions used.
Note that to improve readability, we omit the parameters $\G,\targetset,\lb$ and $\ub$ of the functions in the algorithm. 

\vspace{1ex}
\noindent\textbf{Bounded value iteration:} 
For the standard bounded value iteration algorithm, Line \ref{line:relevantStates} does not run a simulation, but just assigns the whole state space $\states$ to $X$.%
\footnote{Since we mainly talk about simulation based algorithms, we included this line to make their structure clearer.}
Then it updates all values according to the Bellman equations.
After that it finds all the problematic components, the MSECs, and ``deflates'' them as described in~\cite{KKKW18}, i.e. it reduces their values to ensure the convergence to the least fixpoint. 
This suffices for the bounds to converge and the algorithm to terminate~\cite[Theorem 2]{KKKW18}.

\vspace{1ex}
\noindent\textbf{Asynchronous bounded value iteration:}
To tackle the state space explosion problem, \emph{asynchronous} simulation/learning-based algorithms have been developed \cite{BRTDP,BCC+14,KKKW18}. 
The idea is not to update and deflate all states at once, since there might be too many, or since we only have limited information.
%The mentioned approaches are simulation/learning-based.
Instead of considering the whole state space, a path through the SG is sampled by picking in every state one of the actions that look optimal according to the current over-/under-approximation and then sampling a successor of that action. 
This is repeated until either a target is found, or until the simulation is looping in an EC; the latter case occurs if the heuristic that picks the actions generates a pair of strategies under which both players only pick staying actions in an EC. 
After the simulation, only the bounds of the states on the path are updated and deflated.
Since we pick actions which look optimal in the simulation, we almost surely find an $\epsilon$-optimal strategy and the algorithm terminates~\cite[Theorem 3]{BCC+14}. %\todo{"Reverse and add essentially"}

\section{Algorithm}

\subsection{Model-based}%: Counting occurrences}
Given only limited information, updating cannot be done using $\trans$, since the true probabilities are not known.
The approach of \cite{BCC+14} is to sample for a high number of steps and accumulate the observed lower and upper bounds on the true value function for each state-action pair.
When the number of samples is large enough, the average of the accumulator is used as the new estimate for the state-action pair, and thus the approximations can be improved and the results back-propagated, while giving statistical guarantees that each update was correct. 
However, this approach has several drawbacks, the biggest of which is that the number of steps before an update can occur is infeasibly large, often larger than the age of the universe, see Table \ref{tab:main-result} in Section \ref{sec:experiments}.

Our improvements to make the algorithm practically usable are linked to constructing a partial model of the given system.
That way, we have more information available on which we can base our estimates, and we can be less conservative when giving bounds on the possible errors.
The shift from model-free to model-based learning asymptotically increases the memory requirements from $\mathcal{O}(\abs{\states}\cdot \abs{\actions})$ (as in \cite{Strehl,BCC+14}) to $\mathcal{O}(\abs{\states}^2\cdot \abs{\actions})$. 
However, for systems where each action has a small constant bound on the number of successors, which is typical for many practical systems, e.g. classical PRISM benchmarks, it is still $\mathcal{O}(\abs{\states}\cdot \abs{\actions})$ with a negligible constant difference.

We thus track the number of times some successor $\tsucc$ has been observed when playing action $\action$ from state $\state$ in a variable $\#(\state,\action,\tsucc)$.
This implicitly induces the number of times each state-action pair $(\state,\action)$ has been played $\#(\state,\action)= \sum_{\tsucc \in \states} \#(\state,\action,\tsucc)$.
Given these numbers we can then calculate probability estimates for every transition as described in the next subsection.
They also induce the set of all states visited so far, allowing us to construct a partial model of the game.
See \arxivcite{Appendix \ref{app:simulate}}{\cite[Appendix A.2]{arxiv}} for the pseudo-code of how to count the occurrences during the~simulations.

\subsection{Safe updates with confidence intervals using distributed error probability}\label{sec:safeUpdate}

We use the counters to compute a lower estimate of the transition probability for some error tolerance $\transdelta$ as follows:
We view sampling $\tsucc$ from state-action pair $(\state,\action)$ as a Bernoulli sequence, with success probability $\trans(\state,\action,\tsucc)$, the number of trials $\#(\state,\action)$ and the number of successes $\#(\state,\action,\tsucc)$.
The tightest lower estimate we can give using the Hoeffding bound (see \arxivcite{Appendix \ref{app:hoeffding}}{\cite[Appendix D.1]{arxiv}}) is
\begin{equation}\label{eq:transEst}
\widehat\trans(\state,\action,\tsucc) \eqdef \max(0,\frac{\#(s,a,t)}{\#(s,a)} - c),
\end{equation}
where the confidence width $c \eqdef \sqrt{\frac{\ln(\transdelta)}{-2 \#(\state,\action)}}$.
Since $c$ could be greater than 1, we limit the lower estimate to be at least 0.
Now we can give modified update equations:

\begin{align*}
    \widehat\lb(\state,\action) &\eqdef \sum_{\tsucc: \#(\state,\action,\tsucc)>0} \widehat\trans(\state,\action,\tsucc) \cdot \lb(\tsucc)\\
    \widehat\ub(\state,\action) &\eqdef \left( \sum_{\tsucc: \#(\state,\action,\tsucc)>0} \widehat\trans(\state,\action,\tsucc) \cdot \ub(\tsucc) \right)
    +
    \left(1 - \sum_{\tsucc: \#(\state,\action,\tsucc)>0} \widehat\trans(\state,\action,\tsucc)\right)\label{eq:Uest}
\end{align*}

The idea is the same for both upper and lower bound: 
In contrast to the usual Bellman equation (see Section \ref{sec:prelimReach})
%or Algorithm \ref{alg:updateStandard}) 
we use $\widehat\trans$ instead of $\trans$. 
But since the sum of all the lower estimates does not add up to one, there is some remaining probability for which we need to under-/over-approximate the value it can achieve.
We use the safe approximations 0 and 1 for the lower and upper bound respectively; this is why in $\widehat\lb$ there is no second term and in $\widehat\ub$ the whole remaining probability is added.
%this is in fact the sum of all confidence widths of the transitions, or one minus the sum of all lower estimates. 
Algorithm \ref{alg:updateNew} shows the modified update that uses the lower estimates; the proof of its correctness is in \arxivcite{Appendix \ref{app:proofUpdate}}{\cite[Appendix D.2]{arxiv}} .

\begin{lemma}[$\UPDATE$ is correct]\label{lem:updateCorr}
Given correct under- and over-approximations $\lb,\ub$ of the value function $\val$,
and correct lower probability estimates $\widehat\trans$,
the under- and over-approximations after an application of $\UPDATE$ are also correct.
\end{lemma}

\begin{algorithm}[htbp]
	\caption{New update procedure using the probability estimates}\label{alg:updateNew}
	\begin{algorithmic}[1]
		\Procedure{$\UPDATE$}{State set $X$}
		\For {$f \in \set{\lb,\ub}$}  ~~~~~~\Comment{For both functions}
    		\For {$\state \in X \setminus \targetset$} ~~~~\Comment{For all non-target states in the given set}
    		\State 
    		    $ f(\state) = 
    		    \begin{cases}
    		    \max_{\action \in \Av(\state)} \highlight{\widehat{f}}(\state,\action) &\mbox{if } \state \in \states<\Box>\\
    		    \min_{\action \in \Av(\state)} \highlight{\widehat{f}}(\state,\action) &\mbox{if } \state \in \states<\circ>
    		    \end{cases}
    		    $ \label{line:update} %~~~~~~\Comment{Apply Bellman update once}
    		\EndFor
    	\EndFor
		\EndProcedure
	\end{algorithmic}
\end{algorithm}

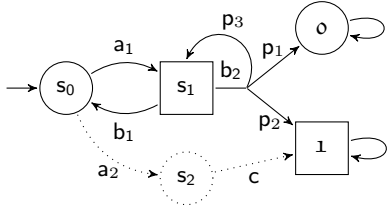
\begin{figure}[t]
\centering
\begin{tikzpicture}[scale=0.8]
%little BCEC with loop
\drawdummy (init) at (0,0) {};
\drawcirc (p) at (1,0) {$\mathsf{s_0}$};
\drawbox (q) at (3,0) {$\mathsf{s_1}$};
\drawdummy (mid) at (4,0) {};
\drawbox (1) at (5.25,-1) {$\target$ };
\drawcirc (0) at (5.25,1)  {$\sink$};
\drawcirc[black,dotted] (r) at (3, -1.5) {$\mathsf{s_2}$};

\draw[->] (init) to (p);
\draw[->]  (p) to[bend left] node [midway,anchor=south] {$\mathsf{a_1}$}(q) ;
\draw[->]  (q) to [bend left] node [midway,anchor=north] {$\mathsf{b_1}$} (p);
\draw[-] (q) to node [midway,anchor=south] {$\mathsf{b_2}$} (mid) ;
\draw[->] (mid) to node [pos=0.5,anchor=south,above] {$\mathsf p_1$} (0);
\draw[->] (mid) to node [pos=0.5,anchor=south,below] {$\mathsf p_2$} (1);
\draw[->] (mid) to [bend left,out=270,in=270,looseness=2] node[above] {$\mathsf p_3$} (q) ;
\draw[->]  (0) to [loop right]  node [midway,anchor=west] {} (0);
\draw[->]  (1) to [loop right] node [midway,anchor=west] {} (1) ;
\draw[->,dotted]  (p) to [bend right] node [midway,anchor=west,below] {$\mathsf a_2$} (r);
\draw[->,dotted]  (r) to node [midway,anchor=west,below] {$\mathsf{c}$} (1);

\end{tikzpicture}
\caption{A running example of an $\SG$.
The dashed part is only relevant for the later examples.
For actions with only one successor, we do not depict the transition probability $1$ (e.g. $\trans(\mathsf{s_0,a_1,s_1})$).
For state-action pair $(\mathsf{s_1,b_2})$, the transition probabilities are parameterized and instantiated in the examples where they are used.
}
	%$\target$ and $\bot$ each have one available action that surely leads back to themselves.}
\label{SGex}
\end{figure}

\begin{example}\label{ex:upd}
We illustrate how the calculation works and its huge advantage over the approach from \cite{BCC+14} on the SG from Figure \ref{SGex}.
For this example, ignore the dashed part and let $\mathsf{p_1}=\mathsf{p_2}=0.5$, i.e. we have no self loop, and an even chance to go to the target $\target$ or a sink $\sink$.
Observe that hence $\val(\mathsf{s_0}) = \val(\mathsf{s_1}) = 0.5$.

Given an error tolerance of $\delta=0.1$,
the algorithm of \cite{BCC+14} would have to sample for more than $10^{9}$ steps before it could attempt a single update.
In contrast, assume we have seen 5 samples of action $\mathsf b_2$, where 1 of them went to $\target$ and 4 of them to $\sink$.
Note that, in a sense, we were unlucky here, as the observed averages are very different from the actual distribution.
The confidence width for $\transdelta = 0.1$ and 5 samples is $\sqrt{\ln(0.1)/-2 \cdot 5} \approx 0.48$.
So given that data, we get $\widehat\trans(\mathsf{s_1,b_2},\target) = \max(0, 0.2 - 0.48) = 0$ and $\widehat\trans(\mathsf{s_1,b_2},\sink) = \max(0, 0.8 - 0.48) = 0.32$. 
Note that both probabilities are in fact lower estimates for their true counterpart.
%Furthermore, although we know that we can reach $\target$, the current estimate for the probability is 0, as there is no larger value that we can guarantee.\todo{in fact, we could use pmin, but we don't.}

Assume we already found out that $\sink$ is a sink with value 0; how we gain this knowledge is explained in the following subsections.
Then, after getting only these 5 samples, $\UPDATE$ already decreases the upper bound of $(\mathsf{s_1,b_2})$ to $0.68$, as we know that at least $0.32$ of $\trans(\mathsf{s_1,b_2})$ goes to the sink.

Given 500 samples of action $\mathsf{b_2}$, the confidence width of the probability estimates already has decreased below $0.05$.
Then, since we have this confidence width for both the upper and the lower bound, we can decrease the total precision for $(\mathsf{s_1,b_2})$ to $0.1$, i.e. return an interval in the order of $[0.45;0.55]$.
\qedtriangle
\end{example} 
Summing up: with the model-based approach we can already start updating after very few steps and get a reasonable level of confidence with a realistic number of samples. 
In contrast, the state-of-the-art approach of \cite{BCC+14} needs a very large number of samples even for this toy example.

Since for $\UPDATE$ we need an error tolerance for every transition, we need to distribute the given total error tolerance $\delta$ over all transitions in the current partial model.
For all states in the explored partial model $\Vis$ we know the number of available actions and can over-approximate the number of successors as $\frac 1 {p_{\min}}$.
Thus the error tolerance for each transition can be set to
$\transdelta \eqdef \frac{\delta \cdot {p_{\min}}}{\abs{\set{\action \mid \state \in \Vis \wedge \action \in \Av(\state)}}}$.
This is illustrated in Example \ref{ex:distrDelta} in \arxivcite{Appendix \ref{app:bigExample}}{\cite[Appendix B]{arxiv}}.

Note that the fact that the error tolerance $\transdelta$ for every transition is the same does \emph{not} imply that the confidence width for every transition is the same, as the latter becomes smaller with increasing number of samples $\#(\state,\action)$.

\subsection{Improved EC detection}
As mentioned in the description of Algorithm \ref{alg:general}, we must detect when the simulation is stuck in a bottom EC %\todo{discuss} 
and looping forever.
However, we may also stop simulations that are looping in some EC but still have a possibility to leave it; for a discussion of different heuristics from \cite{BCC+14,KKKW18}, see \arxivcite{Appendix \ref{app:stuck}}{\cite[Appendix A.3]{arxiv}}.

We choose to define $\STUCK$ as follows: Given a candidate for a bottom EC, we continue sampling until we are $\deltasure$ (i.e. the error probability is smaller than $\transdelta$) that we cannot leave it. 
Then we can safely deflate the EC, i.e. decrease all upper bounds~to~zero.
%\footnote{Note on a corner case: $\DEFLATE$ only decreases to zero if there is no $\targetset$ in the EC.}.

To detect that something is a $\deltasure$ EC, we do not sample for the astronomical number of steps as in \cite{BCC+14}, but rather extend the approach to detect bottom strongly connected components from \cite{DHKP16}.
If in the EC-candidate $T$ there was some state-action pair $(\state,\action)$ that actually has a probability to exit the $T$, that probability is at least $p_{\min}$.
So after sampling $(\state,\action)$ for $n$ times, the probability to overlook such a leaving transition is $(1-p_{\min})^n$ and it should be smaller than $\transdelta$. 
Solving the inequation for the required number of samples $n$ yields
$n \geq \frac{\ln(\transdelta)}{\ln(1-p_{min})}$.

Algorithm \ref{alg:sureEC} checks that we have seen all staying state-action pairs $n$ times, and hence that we are $\deltasure$ that $T$ is an EC.
Note that we restrict to staying state-action pairs, since the requirement for an EC is only that there exist staying actions, not that all actions stay.
We further speed up the EC-detection, because we do not wait for $n$ samples in every simulation, but we use the aggregated counters that are kept over all simulations.

\begin{algorithm}[htbp]
	\caption{Check whether we are $\deltasure$ that $T$ is an EC}\label{alg:sureEC}
	\begin{algorithmic}[1]
		\Procedure{$\highlight{\SUREEC}$}{State set $T$}
		\State $\stepsUntilSure = \frac{\ln(\transdelta)}{\ln(1-p_{\min})}$
		\State $B \gets \set{(\state,\action) \mid \state \in T \wedge \neg (\state,\action) \leaves T}$ \Comment{Set of staying state-action pairs}
		\State \textbf{return} $\bigwedge_{(\state,\action) \in B} \#(\state,\action) > \stepsUntilSure$ 
		\EndProcedure
	\end{algorithmic}
\end{algorithm}

We stop a simulation, if $\STUCK$ returns true, i.e. under the following three conditions:
(i) We have seen the current state before in this simulation ($\state \in X$), i.e. there is a cycle. (ii) This cycle is explainable by an EC $T$ in our current partial model. (iii) We are $\deltasure$ that $T$ is an EC.

\begin{algorithm}[htbp]
	\caption{Check if we are probably looping and should stop the~simulation}\label{alg:stuck}
	\begin{algorithmic}[1]
		\Procedure{$\STUCK$}{State set $X$, state $\state$}
		\If {$\state \notin X$}
		    \State \textbf{return false} \Comment{Easy improvement to avoid overhead}
		\EndIf
		\State \textbf{return } 
		$\highlight{\exists T \subseteq X. T \text{ is EC in partial model} \wedge \state \in T \wedge \SUREEC(T)}$
		\EndProcedure
	\end{algorithmic}
\end{algorithm}

\begin{example}\label{ex:looping}
For this example, we again use the SG from Figure \ref{SGex} without the dashed part, but this time with $\mathsf{p_1}=\mathsf{p_2}=\mathsf{p_3}=\frac 1 3$.
Assume the path we simulated is $(\mathsf{s_0,a_1,s_1,b_2,s_1})$, i.e. we sampled the self-loop of action $\mathsf b_2$.
Then $\set{\mathsf{s_1}}$ is a candidate for an EC, because given our current observation it seems possible that we will continue looping there forever.
However, we do not stop the simulation here, because we are not yet $\deltasure$ about this.
Given $\transdelta = 0.1$, the required samples for that are 6, since $\frac{\ln(0.1)}{\ln(1-\frac{1}{3})} = 5.6$.
With high probability (greater than $(1-\transdelta)=0.9$), within these 6 steps we will sample one of the other successors of $(\mathsf{s_1,b_2})$ and thus realise that we should not stop the simulation in $\mathsf{s_1}$.
If, on the other hand, we are in state $\sink$ or if in state $\mathsf{s_1}$ the guiding heuristic only picks $\mathsf{b_1}$, then we are in fact looping for more than 6 steps, and hence we stop the simulation.
\qedtriangle
\end{example}

\subsection{Adapting to games: Deflating MSECs}\label{sec:games}

To extend the algorithm of \cite{BCC+14} to SGs, instead of collapsing problematic ECs we deflate them as in \cite{KKKW18}, i.e. given an MSEC, we reduce the upper bound of all states in it to the upper bound of the $\exit$ of Maximizer. 
In contrast to \cite{KKKW18}, we cannot use the upper bound of the $\exit$ based on the true probability, but only based on our estimates.
Algorithm \ref{alg:deflate} shows how to deflate an MSEC and highlights the difference, namely that we use $\widehat\ub$ instead of~$\ub$.
\begin{algorithm}[htbp]
	\caption{Black box algorithm to deflate a set of states}\label{alg:deflate}
	\begin{algorithmic}[1]
		\Procedure{$\DEFLATE$}{State set $X$}
		\For {$\state \in X$}
		    \State $\ub(\state) = \min(\ub(\state),\exit(X,\highlight{\widehat\ub})$
		\EndFor
		\EndProcedure
	\end{algorithmic}
\end{algorithm}

The remaining question is how to find MSECs. 
The approach of \cite{KKKW18} is to find MSECs by removing the suboptimal actions of Minimizer according to the current lower bound.
Since it converges to the true value function, all MSECs are eventually found~\cite[Lemma 2]{KKKW18}.
Since Algorithm \ref{alg:find} can only access the SG as a black box, there are two differences:
We can only compare our estimates of the lower bound $\widehat\lb(\state,\action)$ to find out which actions are suboptimal.
Additionally there is the problem that we might overlook an exit from an EC, and hence deflate to some value that is too small; thus we need to check that any state set $\FIND$ returns is a $\deltasure$ EC.
This is illustrated in Example \ref{ex:find}.
For a bigger example of how all the functions we have defined work together, see Example \ref{ex:big} in \arxivcite{Appendix \ref{app:bigExample}}{\cite[Appendix B]{arxiv}}.

\begin{algorithm}[htbp]
	\caption{Finding MSECs in the game restricted to $X$ for black box setting}\label{alg:find}
	\begin{algorithmic}[1]
		\Procedure{$\FIND$}{State set $X$}
		\State $\textit{suboptAct}_\circ \gets \set{(\state,\set{\action \in \Av(\state) \mid \highlight{\widehat\lb}(\state,\action) > \lb(\state)} \mid \state \in \states<\circ> \cap X}$
		\State $\Av' \gets$ $\Av$ without $\textit{suboptAct}_\circ$
		\State $\G[\prime] \gets \G$ restricted to states $X$ and available actions $\Av'$
		\State \textbf{return} $\set{T \in \mec(\G[\prime]) \mid \highlight{\SUREEC(T)}}$
		\EndProcedure
	\end{algorithmic}
\end{algorithm}
%\textcolor{blue}{So in that setting it is only relevant to find the MSECs of the SG eventually in some simulation; when we find them, the correct state set is deflated to get convergence. 
%If we deflate other state sets during the computation, this does not impede correctness.
%Hence the approach of \cite{KKKW18} is to remove those actions of Minimizer that are suboptimal according to the current lower bound $\lb$ and compute the MECs in the remaining game. 
%Since $\lb$ converges to $\val$, the MSECs of the game are detected eventually \cite[Lemma 2]{KKKW18}.}

\begin{example}\label{ex:find}
For this example, we use the full SG from Figure \ref{SGex}, including the dashed part, with $\mathsf{p_1,p_2} > 0$.
Let $(\mathsf{s_0}, \mathsf{a_1}, \mathsf{s_1}, \mathsf{b_2}, \mathsf{s_2}, \mathsf{b_1},\mathsf{s_1}, \mathsf{a_2}, \mathsf{s_2}, \mathsf{c}, \target)$ be the path generated by our simulation.
Then in our partial view of the model, it seems as if $T=\set{\mathsf{s_0,s_1}}$ is an MSEC, since using $\mathsf a_2$ is suboptimal for the minimizing state $\mathsf{s_0}$\footnote{For $\transdelta=0.2$, sampling the path to target once suffices to realize that $\lb(\mathsf{s_0,a_2})>0$.} and according to our current knowledge $\mathsf{a_1,b_1}$ and $\mathsf b_2$ all stay inside $T$.
If we deflated $T$ now, all states would get an upper bound of 0, which would be incorrect.

Thus in Algorithm \ref{alg:find} we need to require that $T$ is an EC $\deltasure\textit{ly}$.
This was not satisfied in the example, as the state-action pairs have not been observed the required number of times.
Thus we do not deflate $T$, and our upper bounds stay correct.
Having seen $(\mathsf{s_1,b_2})$ the required number of times, we probably know that it is exiting $T$ and hence will not make the mistake.
\qedtriangle
\end{example}

\subsection{Guidance and statistical guarantee}\label{sec:full}
It is difficult to give statistical guarantees for the algorithm we have developed so far (i.e. Algorithm \ref{alg:general} calling the new functions from Section \ref{sec:safeUpdate} to \ref{sec:games}). 
Although we can bound the error of each function, applying them repeatedly can add up the error.
Algorithm \ref{alg:blackBVI} shows our approach to get statistical guarantees:
It interleaves a guided simulation phase (Lines \ref{line:beginSim}-\ref{line:endSim}) with a guaranteed standard bounded value iteration (called BVI phase) that uses our new functions (Lines \ref{line:beginBVI}-\ref{line:endBVI}).

The simulation phase builds the partial model by exploring states and remembering the counters.
In the first iteration of the main loop, it chooses actions randomly. 
In all further iterations, it is guided by the bounds that the last BVI phase computed.
After $\theN$ simulations (see below for a discussion of how to choose $\theN$), all the gathered information is used to compute one version of the partial model with probability estimates $\widehat\trans$ for a certain error tolerance $\deltaIter$.
We can continue with the assumption, that these probability estimates are correct, since it is only violated with a probability smaller than our error tolerance (see below for an explanation of the choice of $\deltaIter$).
So in our correct partial model, we re-initialize the lower and upper bound (Line \ref{line:reinit}), and execute a guaranteed standard BVI.
If the simulation phase already gathered enough data, i.e. explored the relevant states and sampled the relevant transitions often enough, this BVI achieves a precision smaller than $\varepsilon$ in the initial state, and the algorithm terminates.
Otherwise we start another simulation phase that is guided by the improved bounds.

%Otherwise we use the bounds calculated by BVI as guidance for the next simulation phase, since they are correct under our assumption that the partial model was correct.
%This way we ensure that the guidance of the simulation depends on sensible values and hence, it explores and samples relevant regions of the state space.
%Thus the next full BVI phase has more information and is more likely to achieve the required precision.

\begin{algorithm}[htbp]
	\caption{Full algorithm for black box  setting}\label{alg:blackBVI}
	\begin{algorithmic}[1]
		\Procedure{BlackVI}{SG $\G$, target set $\targetset$, precision $\varepsilon>0$, error tolerance $\delta>0$}
		\State $\INITIALIZE$
		\State $k = 1$ \Comment {guaranteed BVI counter}
		\State $\Vis \gets \emptyset$ \Comment{current partial model}
		
		\medskip
		
		\Repeat
	    \State $k \gets 2 \cdot k$
	    \State $\deltaIter \gets \frac \delta k$
		\medskip
		\Statex // Guided simulation phase \label{line:beginSim}
		\For {$\theN$ times}
		    \State $X \gets \SIMULATE$
    		\State $\Vis \gets \Vis \cup X$ \label{line:endSim}
    	\EndFor
    	
    	\medskip
    	\Statex // Guaranteed BVI phase 
	    \State $\transdelta \gets \frac{\deltaIter \cdot {p_{\min}}}{\abs{\set{\action \mid \state \in \Vis \wedge \action \in \Av(\state)}}}$ \Comment{Set $\transdelta$ as described in Section \ref{sec:safeUpdate}} \label{line:beginBVI}
		\State $\INITIALIZE$ \label{line:reinit}
		\For {$k \cdot \abs\Vis$ times}\label{line:BVItermcrit}
    	    \State $\UPDATE(\Vis)$ 
    		\For {$T \in \FIND(\Vis)$}
    		    \State $\DEFLATE(T)$\label{line:endBVI}
    		\EndFor
    	\EndFor
    	
    	\Until{$\ub(\initstate) - \lb(\initstate) < \varepsilon$}
		\EndProcedure
	\end{algorithmic}
\end{algorithm}

\paragraph{Choice of $\deltaIter$:}
For each of the full BVI phases, we construct a partial model that is correct with probability $(1-\deltaIter)$. 
%Since we do this repeatedly, the errors can add up.
To ensure that the sum of these errors is not larger than the specified error tolerance $\delta$, we use the variable $k$, which is initialised to 1 and doubled in every iteration of the main loop. Hence for the $i$-th BVI, $k = 2^{i}$.
By setting $\deltaIter = \frac{\delta}{k}$, we get that 
$\displaystyle \sum_{i=1}^{\infty} \deltaIter = \displaystyle \sum_{i=1}^{\infty} \frac{\delta}{2^{i}} = \delta$, and hence the error of all BVI phases does not exceed the specified error tolerance.

\vspace{1ex}
\noindent\textbf{When to stop each BVI-phase:}
The BVI phase might not converge if the probability estimates are not good enough.
We increase the number of iterations for each BVI depending on $k$, because that way we ensure that it eventually is allowed to run long enough to converge. 
On the other hand, since we always run for finitely many iterations, we also ensure that, if we do not have enough information yet, BVI is eventually stopped.
Other stopping criteria could return arbitrarily imprecise results \cite{HM17}.
We also multiply with $\abs{\Vis}$ to improve the chances of the early BVIs to converge, as that number of iterations ensures that every value has been propagated through the whole model at least once.
%When starting a BVI-phase, we do not know, whether we have enough information to achieve precision $\epsilon$; hence it is difficult to specify a termination criterion for the BVI-phase.
%If we stop when the difference between the two most recent approximations is low, we could  return arbitrarily imprecise results \cite{HM17}.
%Thus, by depending on $k$, we use an increasing number of iterations for each BVI, to guarantee that we eventually converge.
%We also multiply $\abs{\Vis}$, as some value might have to be propagated through the whole model.
%By enlarging the number of iterations like this, we can decrease the total number of BVIs.

\vspace{1ex}
\noindent\textbf{Discussion of the choice of $\theN$:}
The number of simulations between the guaranteed BVI phases can be chosen freely;
it can be a constant number every time, or any sequence of natural numbers, possibly parameterised by e.g. $k,$ $\abs\Vis$, $\varepsilon$ or any of the parameters of $\G$.
The design of particularly efficient choices or learning mechanisms that adjust them on the fly is an interesting task left for future work.
We conjecture the answer depends on the given SG and ``task'' that the user has for the algorithm: 
E.g. if one just needs a quick general estimate of the behaviour of the model, a smaller choice of $\theN$ is sensible; 
if on the other hand a definite precision $\varepsilon$ certainly needs to be achieved, a larger choice of $\theN$ is required.

\vspace{-1ex}
\begin{theorem}\label{thm:anyN}
For any choice of sequence for $\theN$, 
Algorithm \ref{alg:blackBVI} is an anytime algorithm with the following property:
When it is stopped, it returns an interval for $\val(\initstate)$ that is PAC \footnote{Probably Approximately Correct, i.e. with probability greater than $1-\delta$, the value lies in the returned interval of width $\varepsilon'$.} for the given error tolerance $\delta$ and some $\varepsilon'$, with $0 \leq \varepsilon' \leq 1$.
\end{theorem}

\vspace{-1ex}
Theorem \ref{thm:anyN} is the foundation of the practical usability of our algorithm. 
Given some time frame and some $\theN$, it calculates an approximation for $\val(\initstate)$ that is probably correct.
Note that the precision $\varepsilon'$ is independent of the input parameter $\varepsilon$, and could in the worst case be always 1.
However, practically it often is good (i.e. close to 0) as seen in the results in Section \ref{sec:experiments}.
Moreover, in our modified algorithm, we can also give a convergence guarantee as in \cite{BCC+14}.
Although mostly out of theoretical interest, in \arxivcite{Appendix~\ref{app:t2}}{\cite[Appendix D.4]{arxiv}} we design such a sequence $\theN$, too.
Since this a-priori sequence has to work in the worst case, it depends on an infeasibly large number of simulations.

\vspace{-1ex}
\begin{theorem}
There exists a choice of $\theN$, such that 
Algorithm \ref{alg:blackBVI} is PAC for any input parameters $\varepsilon, \delta$, i.e. it terminates almost surely 
and returns an interval for $\val(\initstate)$ of width smaller than $\varepsilon$ 
that is correct with probability at least $1-\delta$.
\end{theorem}

\subsection{Utilizing the additional information of grey box input}\label{sec:qual}
In this section, we consider the grey box setting,
i.e. for every state-action pair $(\state,\action)$ we additionally know the exact number of successors $\abs{\post(\state,\action)}$.
Then we can sample every state-action pair until we have seen all successors, and hence this information amounts to having qualitative information about the transitions, i.e. knowing where the transitions go, but not with which probability.

In that setting, we can improve the EC-detection and the estimated bounds in $\UPDATE$.
For EC-detection, note that the whole point of $\SUREEC$ is to check whether there are further transitions available; in grey box, we know this and need not depend on statistics.
%By the qualitative information, we immediately know if there are further transitions for a state-action pair and hence can speed up the process.
For the bounds, note that the equations for $\widehat\lb$ and $\widehat\ub$ both have two parts:
The usual Bellman part and 
the remaining probability multiplied with the most conservative guess of the bound, i.e. 0 and 1.
If we know all successors of a state-action pair, we do not have to be as conservative; then we can use $\min_{\tsucc \in \post(\state,\action)} \lb(\tsucc)$ respectively $\max_{\tsucc \in \post(\state,\action)} \ub(\tsucc)$.
Both these improvements have huge impact, as demonstrated in Section \ref{sec:experiments}.
However, of course, they also assume more knowledge about the model.

\section{Experimental evaluation}\label{sec:experiments}
\begin{table}[t]
\setlength{\tabcolsep}{8pt}
\centering
\caption{Achieved precision $\varepsilon'$ given by our algorithm in both grey and black box settings after running for a period of 30 minutes (See the paragraph below Theorem 1 for why we use $\varepsilon'$ and not $\varepsilon$). 
The first set of the models are MDPs and the second set are SGs. 
`-' indicates that the algorithm did not finish the first simulation phase and hence partial BVI was not called.
$m$ is the number of steps required by the DQL algorithm of \cite{BCC+14} before the first update. 
As this number is very large, we report only $log_{10}(m)$. For comparison, note that the age of the universe is approximately $10^{26}$ nanoseconds; logarithm of number of steps doable in this time is thus in the order of 26.
}
\label{tab:main-result}
\begin{tabular}{r@{\hskip 0.2in}rcllr} % change any of these to H to hide column
\toprule
\multirow[b]{2}{*}{Model} & \multirow[b]{2}{*}{States} & Explored \% & \multicolumn{2}{c}{Precision} & \multirow[b]{2}{*}{$log_{10}(m)$}\\
\cmidrule(lr){4-5}
  &  & Grey/Black & Grey & Black & \\\midrule
consensus          & 272    & 100/100 & 0.00945 & 0.171  & 338     \\
csma-2-2           & 1,038  & 93/93   & 0.00127 & 0.2851 & 1,888   \\
firewire           & 83,153 & 55/-    & 0.0057  & 1      & 129,430 \\
ij-3               & 7      & 100/100 & 0       & 0.0017 & 2,675   \\
ij-10              & 1,023  & 100/100 & 0       & 0.5407 & 17      \\
pacman             & 498    & 18/47   & 0.00058 & 0.0086 & 1,801   \\
philosophers-3     & 956    & 56/21   & 0       & 1      & 2,068   \\
pnueli-zuck-3      & 2,701  & 25/71   & 0       & 0.0285 & 5,844   \\
rabin-3            & 27,766 & 7/4     & 0       & 0.026  & 110,097 \\
wlan-0             & 2,954  & 100/100 & 0       & 0.8667 & 9,947   \\
zeroconf           & 670    & 29/27   & 0.00007 & 0.0586 & 5,998   \\ \midrule
cdmsn              & 1,240  & 100/98  & 0       & 0.8588 & 3,807   \\
cloud-5            & 8,842  & 49/20   & 0.00031 & 0.0487 & 71,484  \\
mdsm-1             & 62,245 & 69/-    & 0.09625 & 1      & 182,517 \\
mdsm-2             & 62,245 & 72/-    & 0.00055 & 1      & 182,517 \\
team-form-3        & 12,476 & 64/-    & 0       & 1      & 54,095  \\ \bottomrule
\end{tabular}
\end{table}

We implemented the approach as an extension of PRISM-Games \cite{PRISM-games}. 11 MDPs with reachability properties were selected from the Quantitative Verification Benchmark Set \cite{qcomp}.
Further, 4 stochastic games benchmarks from \cite{cloud,cdmsn,mdsm,teamform} were also selected.
We ran the experiments on a 40 core Intel Xeon server running at 2.20GHz per core and having 252 GB of RAM.
The tool however utilised only a single core and 1 GB of memory for the model checking.
Each benchmark was ran 10 times with a timeout of 30 minutes. 
We ran two versions of Algorithm \ref{alg:blackBVI}, one with the SG as a black box, the other as a grey box (see Definition \ref{def:limit}).
We chose $\theN = 10,000$ for all iterations.
The tool stopped either when a precision of $10^{-8}$ was obtained or after 30 minutes.
In total, 16 different model-property combinations were tried out.
The results of the experiment are reported in Table \ref{tab:main-result}.

In the black box setting, 
we obtained $\varepsilon < 0.1$ on 6 of the benchmarks. 
    5 benchmarks were `hard' and the algorithm did not improve the precision below 1.
    For 4 of them, it did not even finish the first simulation phase.
    If we decrease $\theN$, the BVI phase is entered, but still no progress is made.
    
In the grey box setting,
on 14 of 16 benchmarks, it took only 6 minutes to achieve $\varepsilon < 0.1$.
    For 8 these, the exact value was found within that time.
    Less than 50\% of the state space was explored in the case of \texttt{pacman}, \texttt{pneuli-zuck-3}, \texttt{rabin-3}, \texttt{zeroconf} and \texttt{cloud\_5}.
    A precision of $\varepsilon < 0.01$ was achieved on 15/16 benchmarks over a period of 30 minutes.

Figure \ref{fig:grey-v-black} shows the evolution of the lower and upper bounds in both the grey- and the black box settings for 4 different models. 
Graphs for the other models as well as more details on the results are in \arxivcite{Appendix \ref{app:experiments}.}{\cite[Appendix C]{arxiv}}.
%\textcolor{blue}{For the SMG \texttt{cloud\_5}, the bounds in the grey box setting converge quickly while in the black box setting, they become 0.1-close within 15 minutes, but don't make progress after that. 
%For the \texttt{pacman} MDP, the bounds converge quite soon in both the settings. 
%In the case of \texttt{rabin}, the grey box result is found in the first partial BVI (indicated by a yellow star) while the black box bounds converge slowly. 
%In the case of \texttt{csma}, the grey box bounds converge fast, but the black box bounds are still far apart.
%The SMG \texttt{mdsm-1} was the worst performer in the grey box setting obtaining an $\epsilon < 0.1$ only after 30-45 minutes. The black box algorithm is unable to make any progress.}

\begin{figure}[t]
    \centering
        \includegraphics[width=0.35\textwidth]{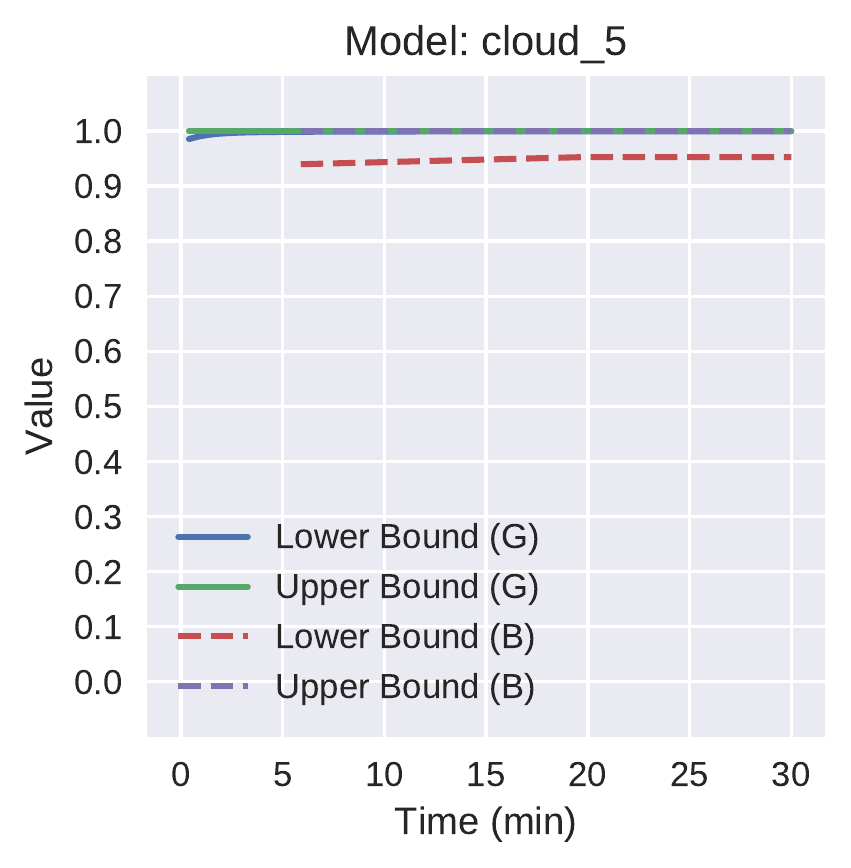}
		\includegraphics[width=0.35\textwidth]{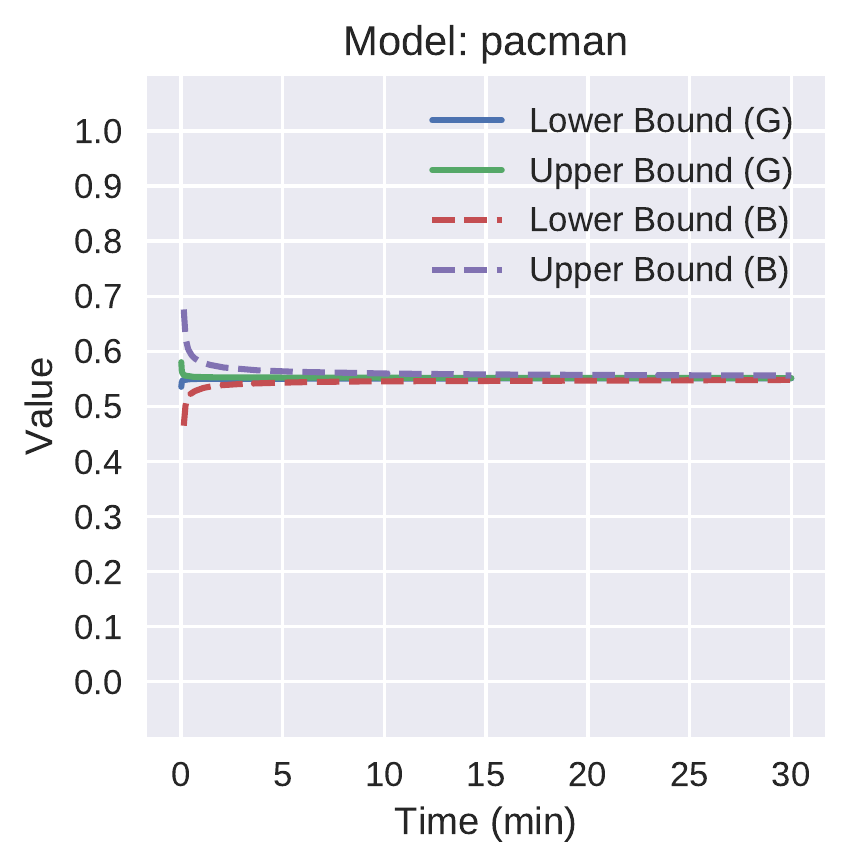}
        
		\includegraphics[width=0.35\textwidth]{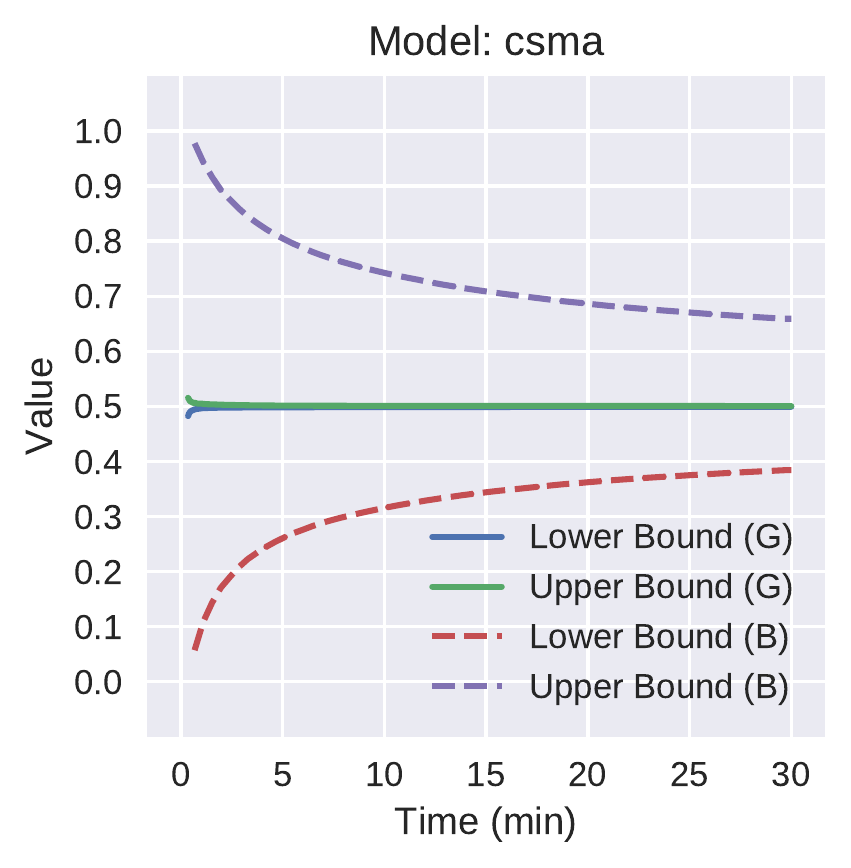}
        \includegraphics[width=0.35\textwidth]{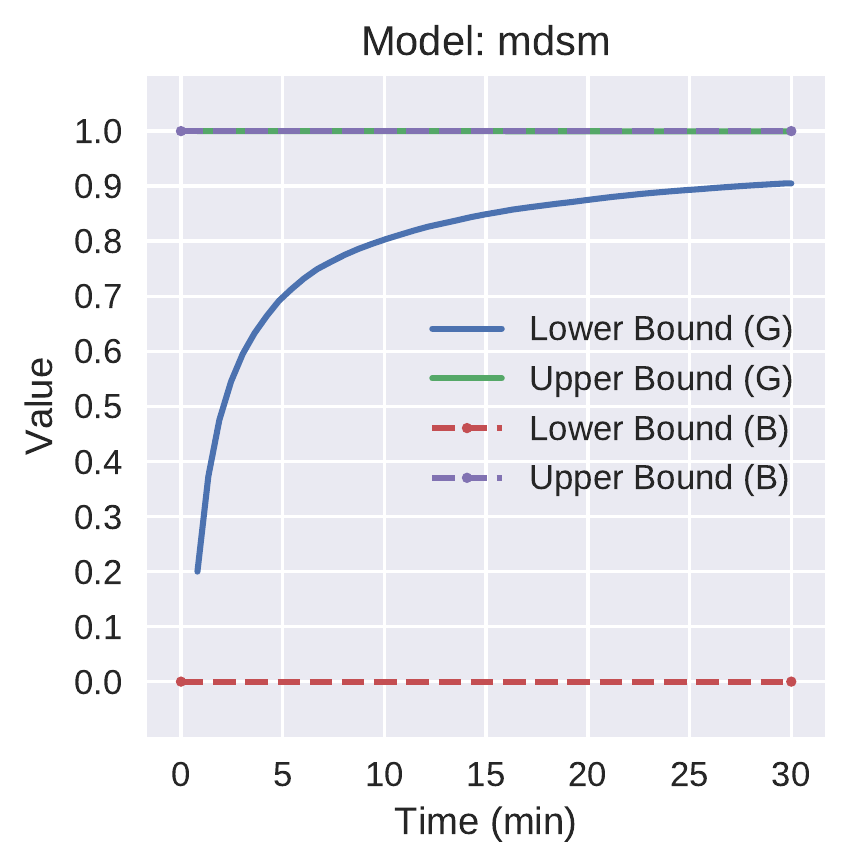}
    \caption{Performance of our algorithm on various MDP and SG benchmarks in grey and black box settings. Solid lines denote the bounds in the grey box setting while dashed lines denote the bounds in the black box setting. The plotted bounds are obtained after each partial BVI phase, because of which they do not start at $[0, 1]$ and not at time 0. Graphs of the remaining benchmarks may be found in \arxivcite{Appendix \ref{app:experiments}.}{\cite[Appendix C]{arxiv}.}}\label{fig:grey-v-black}
\end{figure}

\section{Conclusion}

We presented a PAC SMC algorithm for SG (and MDP) with the reachability objective.
It is the first one for SG and the first practically applicable one.
Nevertheless, there are several possible directions for further improvements.
For instance, one can consider different sequences for lengths of the simulation phases, possibly also dependent on the behaviour observed so far.
Further, the error tolerance could be distributed in a non-uniform way, allowing for fewer visits in rarely visited parts of end components. 
Since many systems are strongly connected, but at the same time feature some infrequent behaviour, this is the next bottleneck to be attacked.~\cite{cores}

%\section*{Acknowledgements}
%This research was funded in part by TUM IGSSE Grant 10.06 (PARSEC), the Czech Science Foundation grant No. 18-11193S, and the German Research Foundation (DFG) project KR 4890/2-1 ``Statistical Unbounded Verification''. We thank Florent Delgrange for his valuable feedback on the proof of Theorem 2.

\pagebreak
\bibliographystyle{alpha}
\bibliography{ref,survey}

\iftoggle{arxiv}{
\appendix
\section*{Appendix}
\section{Pseudocode for the standard algorithms}\label{app:algos}

\subsection{$\INITIALIZE$} \label{app:init}
Initializing the bounds as described in Section \ref{sec:prelimAlgo}.
When simulating, we do not explicitly initialize the bounds for all states up front, but we rather use these values as soon as we encounter a new state.
\begin{algorithm}[htbp]
	\caption{Initializing the lower and upper bounds on the value in the most conservative way}\label{alg:init}
	\begin{algorithmic}[1]
		\Procedure{$\INITIALIZE$}{}
		\For {$\state \in \states$}
		\State $\lb(\state)=
    		\begin{cases}
    		1 &\mbox{if } \state \in \targetset\\
    		0 &\mbox{else}
    		\end{cases}$ ~~~~~~\Comment{Lower bound}
		\State $\ub(\state)=1$ ~~~~~~~~~~~~~~~~\Comment{Upper bound}
		\EndFor
		\EndProcedure
	\end{algorithmic}
\end{algorithm}

\subsection{$\SIMULATE$}\label{app:simulate}

\para{In the full information setting,}
$\SIMULATE$ works as follows:
The set of $\best$ actions for a state $\state$, given the the current $\ub$ and $\lb$, is defined as:
\[\best_{\ub,\lb}(\state) := 
\begin{cases} \set{\action \in \Av(\state) \mid \ub(\state,\action) = \max_{\action \in \Av(\state)}\ub(\state,\action)}		&\mbox{if } \state \in \states<\Box>  \\
\set{\action \in \Av(\state) \mid \lb(\state,\action) =\min_{\action \in \Av(\state)}\lb(\state,\action)} &\mbox{if } \state \in \states<\circ>. \end{cases}\]
So $\SIMULATE$ can be viewed as 
first fixing the strategies $\sigma,\tau$, such that for every state $\state$ they randomize uniformly over all actions in $\best_{\ub,\lb}(\state)$, 
and then sampling in the induced Markov chain $\G[\sigma,\tau]$.

\begin{algorithm}[H]
\caption{Standard simulation algorithm to sample a path}\label{alg:simulate}
\begin{algorithmic}[1]
\Procedure{$\SIMULATE$}{lower bound function $\lb$, upper bound function $\ub$}
    \State $X \gets \emptyset$ \Comment{Set of states visited during this simulation}
    \State $\state \gets \initstate$ 
    \Repeat
    	\State $X \gets X \cup \state$
    	\State $\action \gets $ sampled uniformly from $\best_{\ub,\lb}(\state)$
    	\State $\state \gets$ sampled according to $\trans(\state,\action)$
    \Until{$\state \in \targetset$ or $\STUCK(\state,X)$}
    \State \textbf{return} $X \cup \set{\state}$ \Comment{Add last state before returning path.}
\EndProcedure
\end{algorithmic}
\end{algorithm}

\para{In the limited information setting,} we use a map $\states \times \actions \times \states \to \mathbb{N}$ to store the number of times a triple $(\state,\action,\tsucc)$ has been observed. 
Initially, all counters are set to 0.

Then the only difference between Algorithm \ref{alg:simulate} and \ref{alg:simulateNew} is the highlighted line, where we increment the counter of the observed triple, and the fact that we explicitly save the successor state in a variable to be able to do so.

\begin{algorithm}[H]
\caption{New simulation algorithm counting occurrences}\label{alg:simulateNew}
\begin{algorithmic}[1]
\Procedure{$\SIMULATE_{counting}$}{}
    \State $X \gets \emptyset$ \Comment{Set of states visited during this simulation}
    \State $\state \gets \initstate$ 
    \Repeat
    	\State $X \gets X \cup \state$
    	\State $\action \gets $ sampled uniformly from $\best_{\ub,\lb}(\state)$
    	\State $\highlight{\tsucc} \gets$ sampled according to $\trans(\state,\action)$
    	\State $\highlight{\text{Increment }\#(\state,\action,\tsucc)}$
    	\State $\state \gets \highlight{\tsucc}$
    \Until{$\state \in \targetset$ or $\STUCK_{safe}(\state,X)$}
    \State \textbf{return} $X \cup \set{\state}$ 
\EndProcedure
\end{algorithmic}
\end{algorithm}

\subsection{$\STUCK$}\label{app:stuck}

There are several heuristics for $\STUCK$ described in \cite{BCC+14,KKKW18}.
The requirements for the heuristic to be correct are, that (1) if we are stuck in a bottom EC, it definitely returns true eventually and then we stop the simulation; and (2) if there is positive chance to reach a state in the Markov chain $\G[\sigma,\tau]$ (see Section \ref{app:simulate}), there also is a positive chance that the simulation reaches it. 
This would be violated if the simulation depended on a constant number of steps that is smaller than the length of the longest path in the SG.

Correct heuristics for $\STUCK$ are e.g. 
    stopping a simulation after a certain length of the path is exceeded, where this length either grows with the number of simulations or is larger than the size of the partial model;
or one can stop a simulation as soon as a state appears the second time in the path.
This is the heuristic that we use in Algorithm \ref{alg:stuckStandard}.
One could also be rigorous and check that in fact a bottom EC has been reached and there is no way to exit it. However, this check is computationally costly, and using the heuristics proved to be faster.

\begin{algorithm}[htbp]
	\caption{Heuristic to check whether we are probably looping and hence should stop the simulation}\label{alg:stuckStandard}
	\begin{algorithmic}[1]
		\Procedure{$\STUCK$}{State set $X$, state $\state$}
		\State \textbf{return $\state \notin X$}
		\EndProcedure
	\end{algorithmic}
\end{algorithm}

\subsection{$\UPDATE$}\label{app:update}
Algorithm \ref{alg:updateStandard} is the standard update procedure for the full information setting as used in e.g. \cite{BCC+14,HM17} (without the case distinction on players) and \cite{KKKW18}.

Note that we can only update states in $X \setminus \targetset$, since targets are set correctly already, but updating them might be wrong (if they do not self-loop, but go somewhere else). This problem was not addressed in the other papers because of their preprocessing.

\begin{algorithm}[htbp]
	\caption{Standard update procedure}\label{alg:updateStandard}
	\begin{algorithmic}[1]
		\Procedure{$\UPDATE$}{State set $X$, lower bound function $\lb$, upper bound function $\ub$}
		\For {$f \in \set{\ub,\lb}$}  ~~~~~~\Comment{For both functions}
    		\For {$\state \in X \setminus \targetset$} ~~~~\Comment{For all non-target states in the given set}
    		\State $ f(\state) = 
    		    \begin{cases}
    		    \max_{\action \in \Av(\state)} \sum_{\tsucc \in \post(\state,\action)} \trans(\state,\action)(\tsucc) \cdot f(\state) &\mbox{if } \state \in \states<\Box>\\
    		    \min_{\action \in \Av(\state)} \sum_{\tsucc \in \post(\state,\action)} \trans(\state,\action)(\tsucc) \cdot f(\state) &\mbox{if } \state \in \states<\circ>
    		    \end{cases}
    		    $  %~~~~~~\Comment{Apply Bellman update once}
    		\EndFor
    	\EndFor
		\EndProcedure
	\end{algorithmic}
\end{algorithm}

\subsection{$\DEFLATE$ and $\FIND$}\label{app:games}
This section contains the algorithms for finding MSECs and deflating them from \cite{KKKW18}.

Algorithm \ref{alg:findStandard} finds all MSECs in $X$ in the full information setting. 
It works by removing the suboptimal actions of Minimizer and computing MECs in the thus restricted game.
Any MEC that uses only optimal actions of Minimizer is an MSEC according to $\lb$~\cite[Lemma 2]{KKKW18}.
Note that for convergence of the full algorithm we use that $\lb$ converges towards $\val$, so eventually all decisions of Minimizer are set correctly and we find the true MSECs.

\begin{algorithm}[htbp]
	\caption{Algorithm to find all MSECs in the game restricted to $X$}\label{alg:findStandard}
	\begin{algorithmic}[1]
		\Procedure{$\FIND$}{State set $X$}
		\State $\textit{suboptAct}_\circ \gets \set{(\state,\set{\action \in \Av(\state) \mid \widehat\lb(\state,\action) > \widehat\lb(\state)} \mid \state \in \states<\circ> \cap X}$
		\State $\G[\prime] \gets \G$ with $\Av$ replaced by $\Av$ without $\textit{suboptAct}_\circ$
		\State \textbf{return} $\mec(\G[\prime])$
		\EndProcedure
	\end{algorithmic}
\end{algorithm}

Algorithm \ref{alg:deflateStandard} shows how to adjust the upper bounds in an MSEC to ensure convergence, i.e. avoid the over-approximation being stuck at a greater fixpoint than the least.
The resulting upper bound actually is sound for any set of states $X$ given as input~\cite[Lemma 3]{KKKW18}.

\begin{algorithm}[htbp]
	\caption{Algorithm to deflate a set of states}\label{alg:deflateStandard}
	\begin{algorithmic}[1]
		\Procedure{$\DEFLATE$}{State set $X$}
		\For {$\state \in X$}
		    \State $\ub(\state) = \min(\ub(\state),\exit(X,\ub))$
		\EndFor
		\EndProcedure
	\end{algorithmic}
\end{algorithm}

\section{Additional examples}\label{app:bigExample}
\begin{example}\label{ex:distrDelta}
We use the same SG as in Example \ref{ex:upd} and a total error tolerance of $\delta=0.1$.
We count the number of state-action pairs in our partial model as 6, but still have to over-estimate that every one of them has $\frac 1 {p_{\min}} = \frac 1 {0.5} = 2$ successors, so then $\transdelta = \frac {0.1} {12} = 0.008$.

Note that in fact the only state-action pair in this example, where we really needed this error tolerance, is $(\mathsf{s_1,b_2})$. 
Any state-action pair that does not have $\frac 1 {p_{\min}}$ successors reduces the true number of transitions that would need a part of the error tolerance. 
Hence, for realistic models we often still distribute very conservatively.
\qedtriangle
\end{example}

\begin{example}\label{ex:big}
In this example we illustrate how all our functions work together.
We use the full SG from Figure \ref{SGex}, including the dashed part, and set $\mathsf{p_1}=\mathsf{p_2}=0.5$.

Our first simulation reached $\sink$. 
We stayed there until $\STUCK$ returned true, since we looped the required number of times.
$\set{\sink}$ is a MEC in the game restricted to optimal actions of Minimizer and the states explored in the first simulation, and we already ensured it is an EC $\deltasure\text{ly}$.
Hence it is in the set returned by $\FIND$, and $\ub(\sink)$ is set to 0.

Now we simulate enough times, such that we just realized that the upper bound of $(\mathsf{s_1,b_2})$ is something smaller than 1, since we know there is some probability to go to the sink $\sink$ (similar to Example \ref{ex:upd}).
Additionally, we also already know that $(\mathsf{s_0,a_2})$ has a positive probability of reaching the target, a higher one than $(\mathsf{s_1,b_2})$(by simulating paths as in Example \ref{ex:find}).

Then $T=\set{\mathsf{s_0,s_1}}$ forms an MSEC, since it is Minimizer's best choice to remain in $T$ and Maximizer is under the illusion that going to $\mathsf{s_0}$ promises a value of 1, since the $\ub(\mathsf{s_0})$ still is 1; hence the heuristic picks $\mathsf b_1$, as it already knows that $\mathsf b_2$ yields a value smaller than 1.
Thus, the next simulation loops in $T$ until we know that it is an EC $\deltasure\text{ly}$ and thus we stop simulating.
Then deflate decreases the upper bound of all states in $T$, i.e. of $\mathsf{s_0}$ and $\mathsf{s_1}$, to $\widehat\ub(\mathsf{s_1,b_2})$, which is the best thing that Maximizer can achieve in $T$, and the least thing that Minimizer must allow to happen.

Then, in the next simulation, $\mathsf{s_1}$ randomizes uniformly between $\mathsf{b_1}$ and $\mathsf{b_2}$, and hence there is a positive chance of sampling the state-action pair $(\mathsf{s_1,b_2})$ and improving our probability estimate $\widehat\trans(\mathsf{s_1,b_2})$.
Thus, we eventually improve our bounds for $(\mathsf{s_1,b_2})$.

Then the process is repeated, i.e. we are stuck in $T$ again, which we realize immediately this time, since we can access the information from the previous simulations;
then we deflate the upper bounds in $T$ to $\widehat\ub(\mathsf{s_1,b_2})$ and after that again eventually improve the bounds for $(\mathsf{s_1,b_2})$.
This way, we eventually are able to achieve any precision $\epsilon$, since we know all relevant states (actually all states, in this example) and the estimate $\widehat\trans(\mathsf{s_1,b_2})$ becomes more and more precise given more and more samples.

However, as described in Section \ref{sec:full} doing it like this might add up the error, which is why we introduce the two phase approach.
\qedtriangle
\end{example}

\newpage
\section{Experimental Results} \label{app:experiments}

\begin{figure}[ht]
    \centering
    \includegraphics[width=0.48\textwidth]{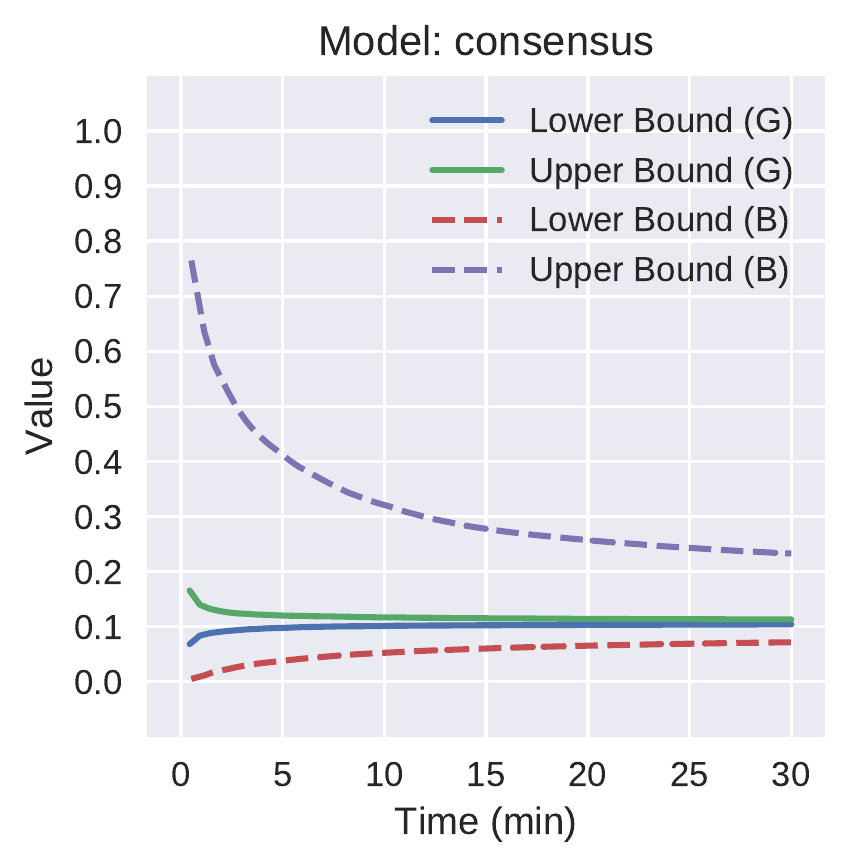}
    \includegraphics[width=0.48\textwidth]{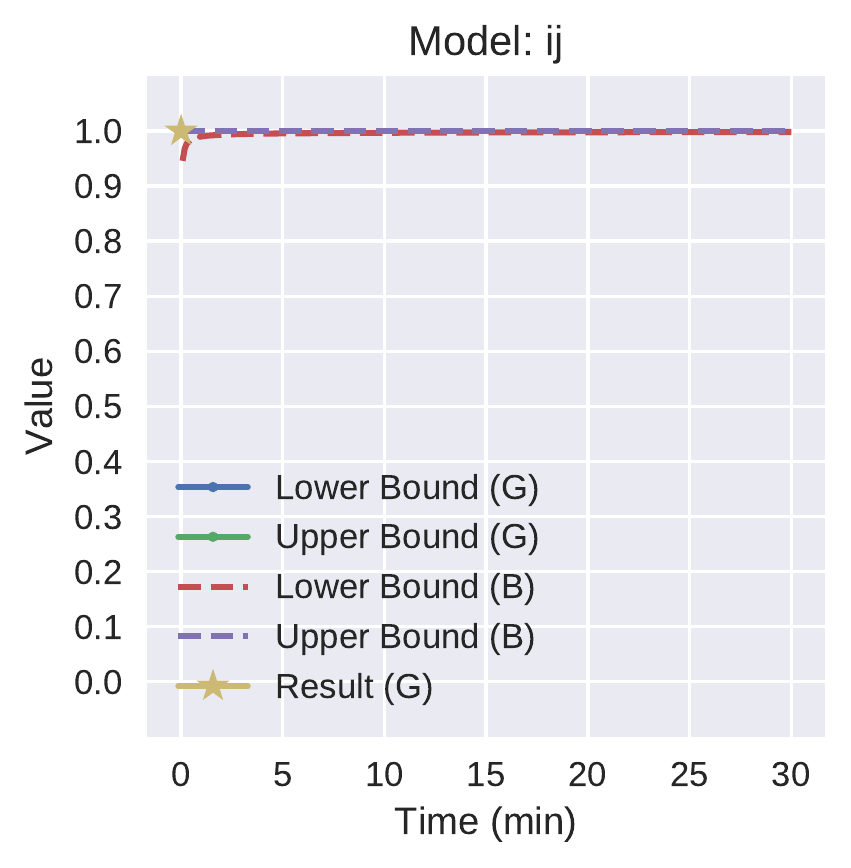}
    
    \includegraphics[width=0.48\textwidth]{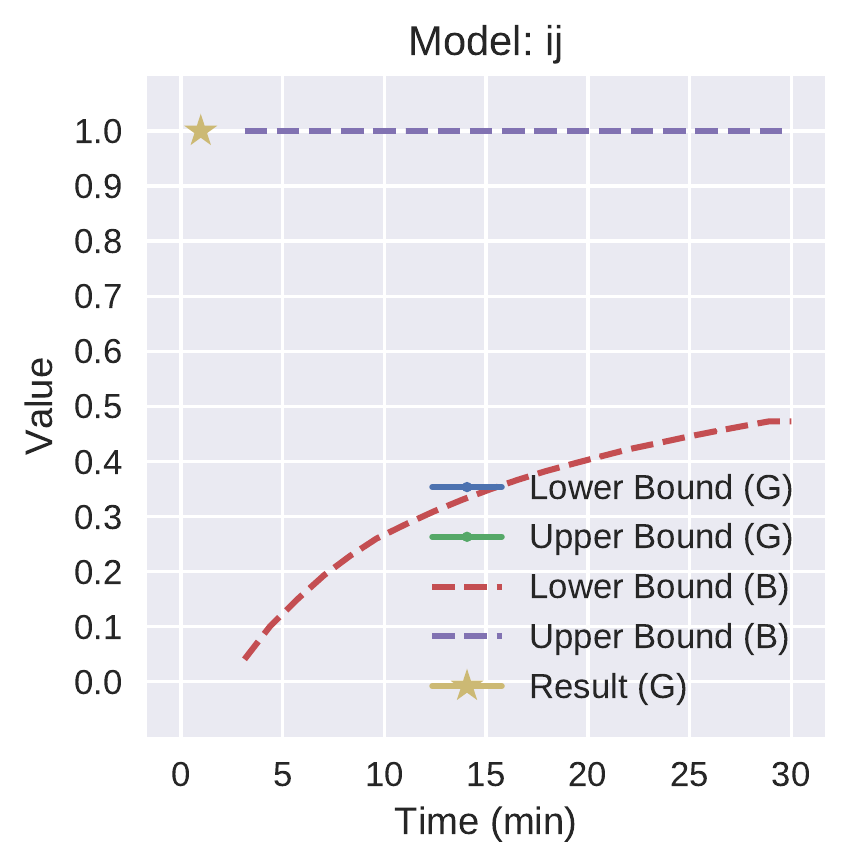}
    \includegraphics[width=0.48\textwidth]{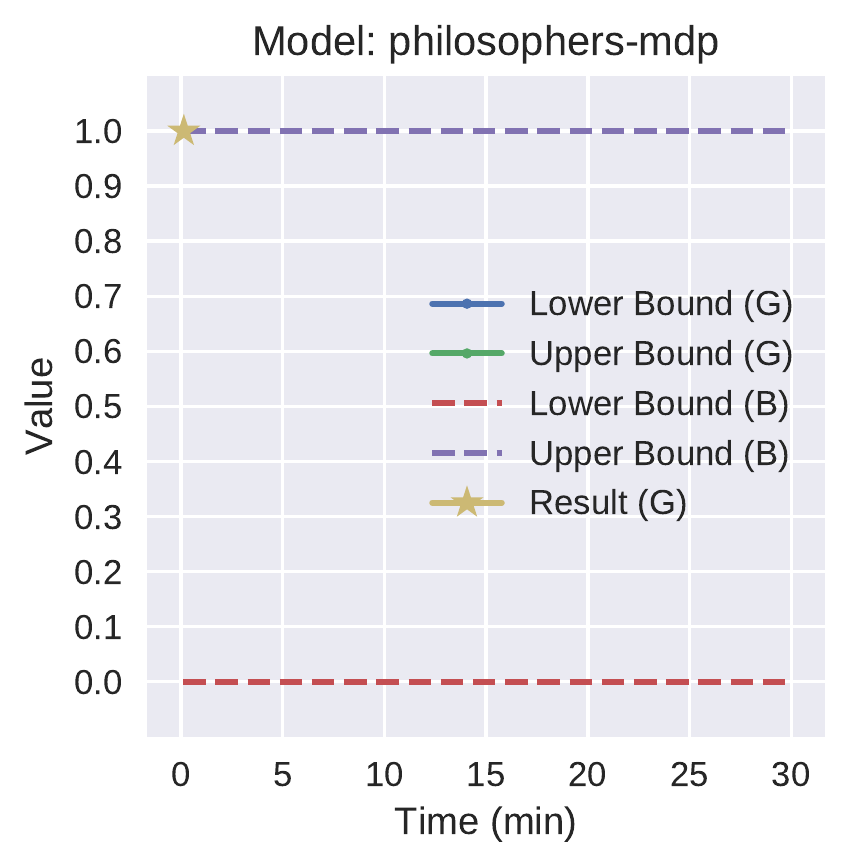}
    \caption{Performance of our algorithm on various MDP and SMG benchmarks in grey- and black box settings. Solid lines denote the bounds in the grey box setting while dashed lines denote the bounds in the black box setting. A star denotes that the algorithm terminated after the difference in bounds became less than $10^{-8}$.}
\end{figure}
\begin{figure}[ht]
    \centering
    \includegraphics[width=0.45\textwidth]{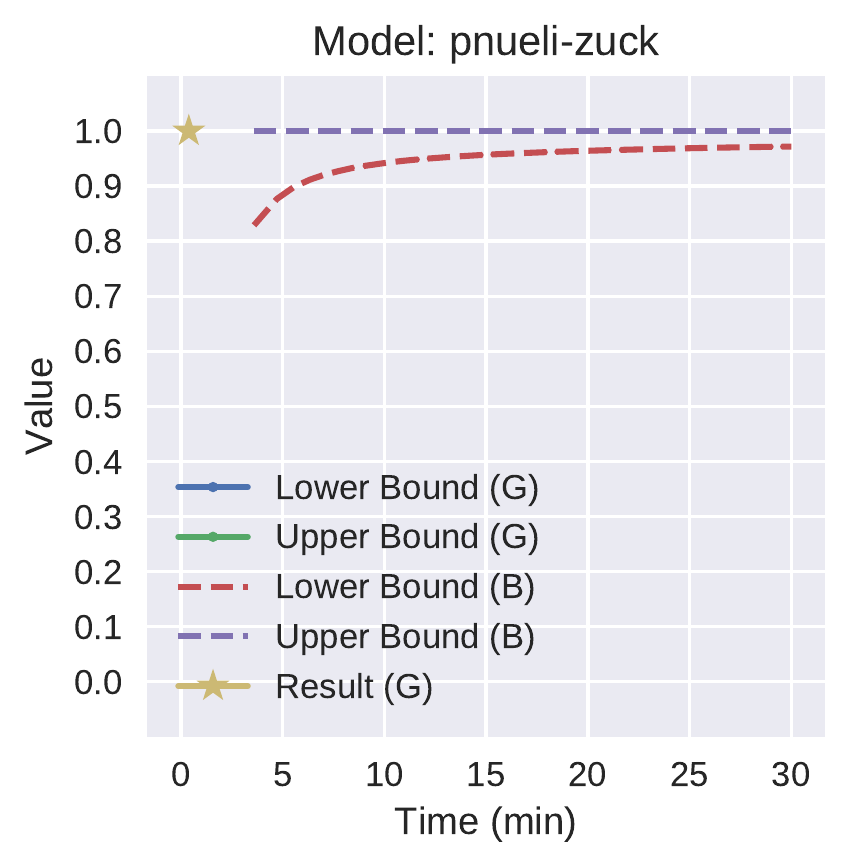}
    \includegraphics[width=0.45\textwidth]{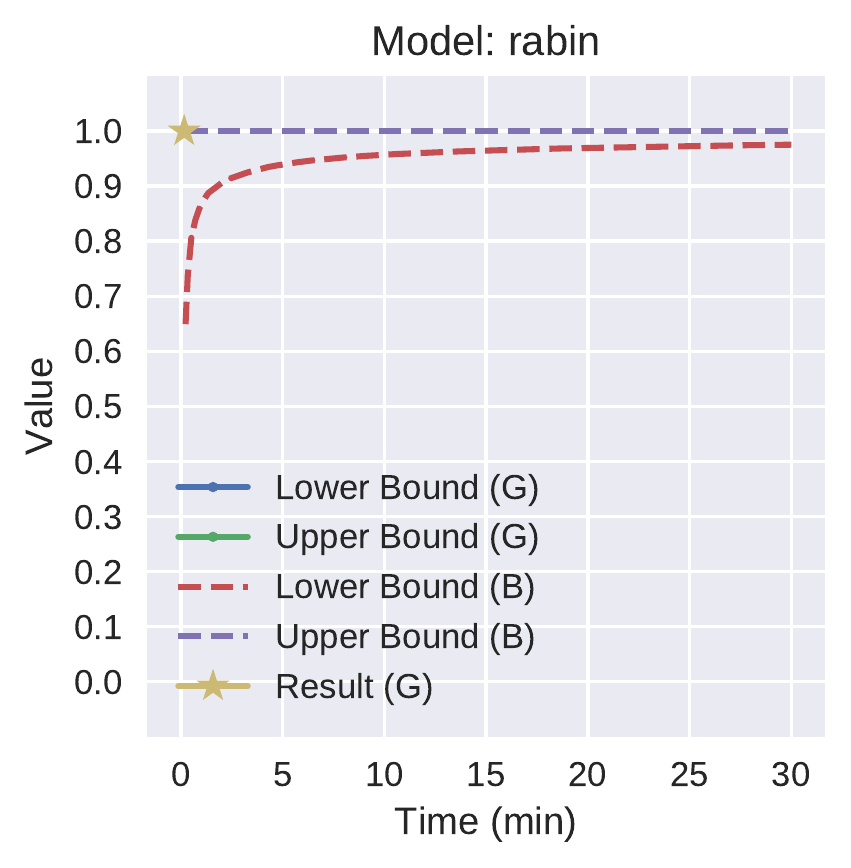}
    
    \includegraphics[width=0.45\textwidth]{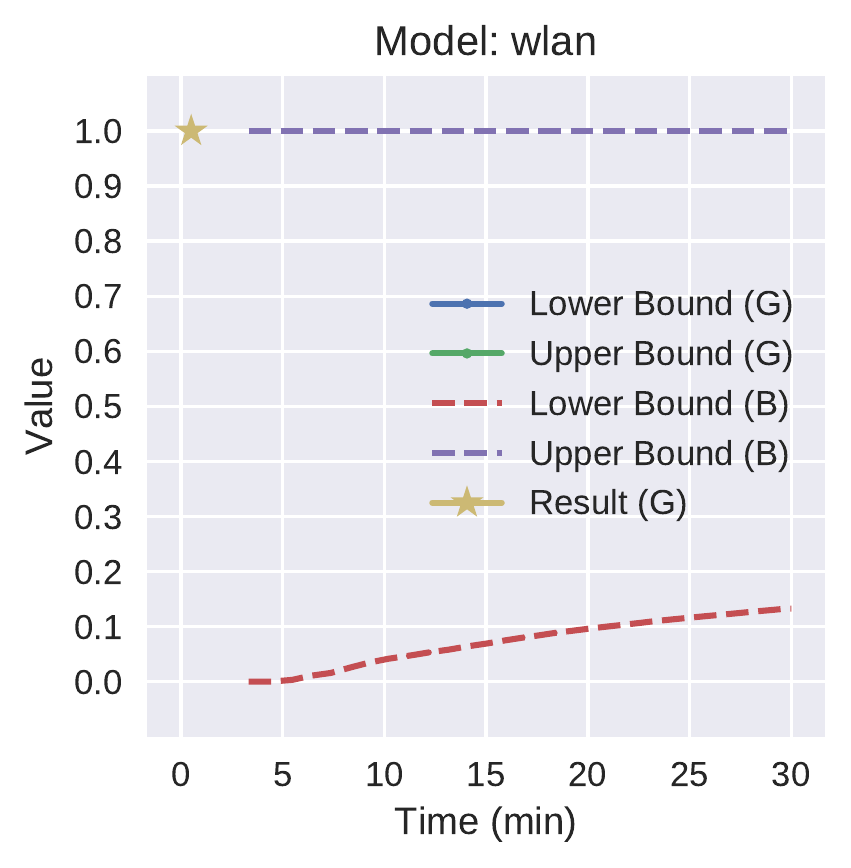}
    \includegraphics[width=0.45\textwidth]{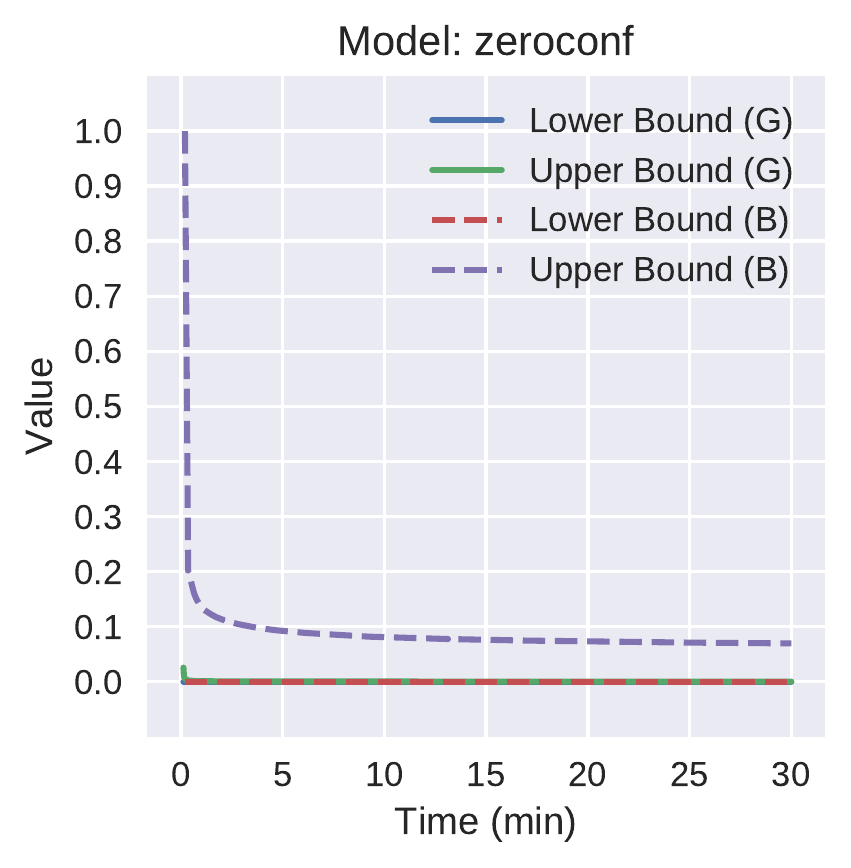}
    
    \includegraphics[width=0.45\textwidth]{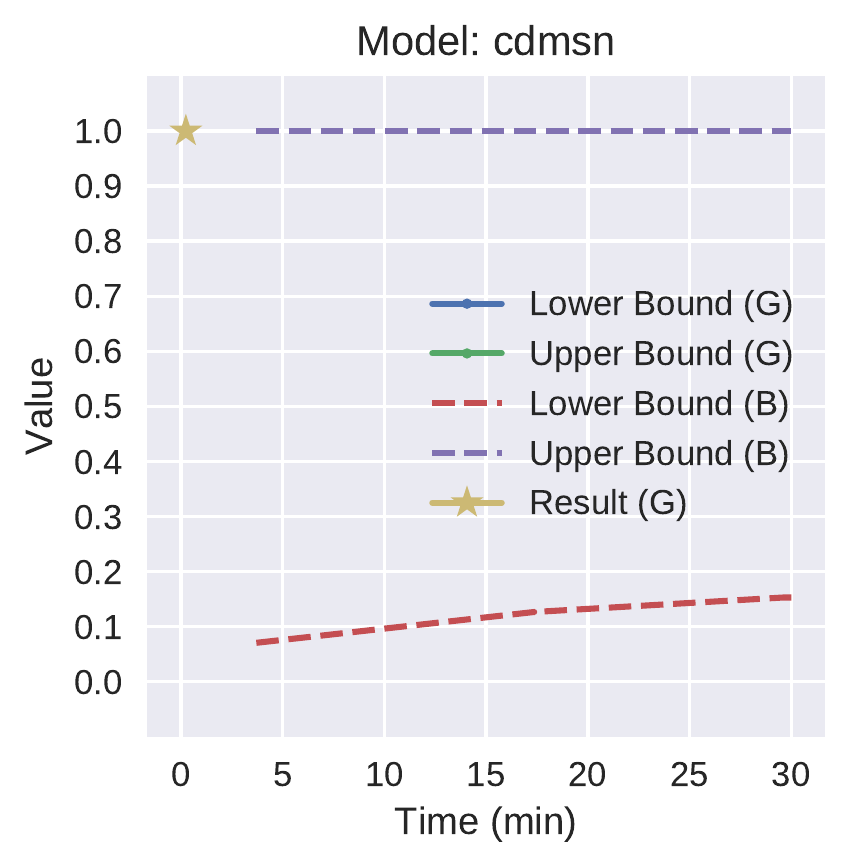}
    \includegraphics[width=0.45\textwidth]{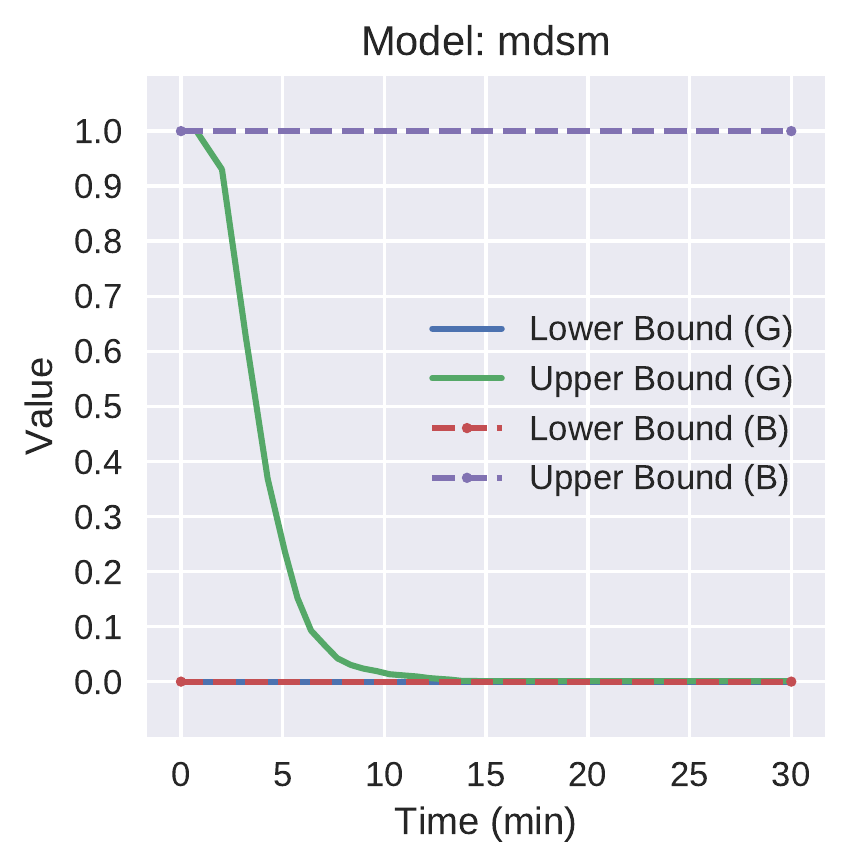}
    \caption{Performance of our algorithm on various MDP and SG benchmarks in grey- and black box settings. Solid lines denote the bounds in the grey box setting while dashed lines denote the bounds in the black box setting. A star denotes that the algorithm terminated after the difference in bounds became less than $10^{-8}$.}
\end{figure}

% Please add the following required packages to your document preamble:
% \usepackage{multirow}
% \usepackage{lscape}
\begin{landscape}
\begin{table}[]
\setlength{\tabcolsep}{8pt}
\centering
\caption{A full version of Table \ref{tab:main-result}} \label{app:tab:complete-results-table}
\begin{tabular}{rlrlcllccl}
\toprule
\multirow{2}{*}{Model} & \multirow{2}{*}{Type} & \multirow{2}{*}{States} & \multirow{2}{*}{Property} & \multirow{2}{*}{Explored \%} & \multicolumn{2}{c}{Precision} & \multirow{2}{*}{Grey Bounds} & \multirow{2}{*}{Black Bounds} & \multirow{2}{*}{True Value}\\
\cmidrule(lr){6-7}
                   &      &        &                     & Grey/Black & Grey & Black  & Lower/Upper   & Lower/Upper   &        \\ \midrule
brp                & DTMC & 677    & p1                  & 95/91   & 0.00233 & 0.1937 & 0/0.0023      & 0/0.1937      & 0.0004 \\
crowds             & DTMC & 1,198  & positive            & 96/96   & 0.00558 & 0.3154 & 0.0507/0.0563 & 0.0412/0.3567 & 0.0529 \\
haddad-monmege     & DTMC & 41     & targt               & 100/-   & 0.87843 & 1      & 0.1086/0.987  & 0/1           & 0.7    \\
leader-sync-3-2    & DTMC & 26     & eventually\_elected & 100/100 & 0       & 0.0037 & 1/1           & 0.9962/1      & 1      \\ \midrule
consensus          & MDP  & 272    & disagree            & 100/100 & 0.00945 & 0.171  & 0.1039/0.1133 & 0.0694/0.2404 & 0.1083 \\
csma-2-2           & MDP  & 1,038  & some\_before        & 93/93   & 0.00127 & 0.2851 & 0.4993/0.5005 & 0.3792/0.6643 & 0.5    \\
firewire           & MDP  & 83,153 & deadline            & 55/-    & 0.0057  & 1      & 0.4968/0.5025 & 0/1           & 0.5    \\
ij-10              & MDP  & 1,023  & stable              & 100/100 & 0       & 0.5407 & 1/1           & 0.4592/1      & 1      \\
ij-3               & MDP  & 7      & stable              & 100/100 & 0       & 0.0017 & 1/1           & 0.9982/1      & 1      \\
pacman             & MDP  & 498    & crsh                & 18/47   & 0.00058 & 0.0086 & 0.5508/0.5514 & 0.5477/0.5564 & 0.5511 \\
philosophers-mdp-3 & MDP  & 956    & eat                 & 56/21   & 0       & 1      & 1/1           & 0/1           & 1      \\
pnueli-zuck-3      & MDP  & 2,701  & live                & 25/71   & 0       & 0.0285 & 1/1           & 0.9714/1      & 1      \\
rabin-3            & MDP  & 27,766 & live                & 7/4     & 0       & 0.026  & 1/1           & 0.9739/1      & 1      \\
wlan-0             & MDP  & 2,954  & sent                & 100/100 & 0       & 0.8667 & 1/1           & 0.1332/1      & 1      \\
zeroconf           & MDP  & 670    & correct\_max        & 29/27   & 0.00007 & 0.0586 & 0/0           & 0/0.0586      & 0      \\ \midrule
cdmsn              & SG  & 1,240  & all\_prefer\_one    & 100/98  & 0       & 0.8588 & 1/1           & 0.1411/1      & 1      \\
cloud-5            & SG  & 8,842  & eventually\_deploy  & 49/20   & 0.00031 & 0.0487 & 0.9996/1      & 0.9512/1      & 0.9999 \\
mdsm               & SG  & 62,245 & player\_1\_deviate  & 69/-    & 0.09625 & 1      & 0.9031/0.9993 & 0/1           & 0.989  \\
mdsm               & SG  & 62,245 & player\_2\_deviate  & 72/-    & 0.00055 & 1      & 0/0.0005      & 0/1           & 0      \\
team-form-3        & SG  & 12,476 & completed           & 64/-    & 0       & 1      & 1/1           & 0/1           & 1      \\ \bottomrule
\end{tabular}
\end{table}
\end{landscape}

\section{Proofs}

\subsection{Proof that $\widehat{\trans}$ is the best estimate possible by using the Hoeffding bound}\label{app:hoeffding}

We view sampling $\tsucc$ from state-action pair $(\state,\action)$ as a Bernoulli sequence,  with success probability $\trans(\state,\action,\tsucc)$ and the number of trials $\#(\state,\action)$.
Given the number of successes $\#(\state,\action,\tsucc)$, we want to find a lower estimate $\widehat\trans$, such that $\Prob(\widehat\trans \geq \trans(\state,\action,\tsucc)) \leq \transdelta$.

\begin{align*}
    \Prob \left(\widehat\trans(\state,\action,\tsucc) \geq \trans(\state,\action,\tsucc)\right) 
    &= \Prob\left(\frac{\#(s,a,t)}{\#(s,a)} - c \geq \trans(\state,\action,\tsucc)\right) \tag{Equation \ref{eq:transEst}}\\
    &\leq e^{-2c^2 \#(\state,\action)} \tag{\cite[Theorem 1]{hoeffding}}\\
    &\leq \transdelta \tag{Required bound on error probability}
\end{align*}
Solving the last two lines for the confidence width $c$ yields $c \geq \sqrt{\frac{\ln(\transdelta)}{-2 \#(\state,\action)}}$.

\subsection{Proof of Lemma \ref{lem:updateCorr}}\label{app:proofUpdate}
\setcounter{lemma}{0}
We slightly reformulate the lemma.
\begin{lemma}[$\UPDATE$ is correct]
Let $\lb$ and $\ub$ be correct under- respectively over-approximations of the value function $\val$, i.e. for all states $\state$ it holds that $\lb(\state) \leq \val(\state) \leq \ub(\state)$; 
and let $\lb'$ and $\ub'$ be the new estimates after an application of $\UPDATE$.

Assuming that $\widehat\trans(\state,\action,\tsucc) \leq \trans(\state,\action,\tsucc)$, $\UPDATE$ is correct, i.e. for all states $\state$ it holds that $\lb'(\state) \leq \val(\state) \leq \ub'(\state)$
\end{lemma}
\begin{proof}
We prove, that actually for every action $\action \in \Av(\state)$ the claim holds.
Then trivially it also holds for the whole state, independent of player, since if the estimates for every action are correct, also the maximum/minimum estimate is correct.

We exemplify the proof for the upper bound; the lower bound is analogous.
\begin{align*}
    \val(\state,\action) &\leq 
    \sum_{\tsucc \in \post(\state,\action)} \trans(\state,\action,\tsucc) \cdot \ub(\tsucc) 
        \tag{Bellman equation and $\ub$ is correct over-approximation}\\
    &=
    \sum_{\tsucc \in \post(\state,\action)} 
        \widehat\trans(\state,\action,\tsucc) \cdot \ub(\tsucc)
        +
        (\trans(\state,\action,\tsucc) - \widehat\trans(\state,\action,\tsucc)) \cdot \ub(\tsucc)
        \tag{splitting $\trans$ into the estimate and the remaining probability}\\
    &\leq 
    \sum_{\tsucc \in \post(\state,\action)} 
        \widehat\trans(\state,\action,\tsucc) \cdot \ub(\tsucc)
        +
        (\trans(\state,\action,\tsucc) - \widehat\trans(\state,\action,\tsucc)) \cdot 1
        \tag{1 is the maximal value, so $\ub(\tsucc) \leq 1$}\\
    &=
    \left(\sum_{\tsucc \in \post(\state,\action)} \widehat\trans(\state,\action,\tsucc) \cdot \ub(\tsucc)\right)
    + 
    \left(\sum_{\tsucc \in \post(\state,\action)} \trans(\state,\action,\tsucc)\right)\\
    &~~~~~~~~~~~~~~~~~~~~~~~~~~~ -
    \left(\sum_{\tsucc \in \post(\state,\action)} \widehat\trans(\state,\action,\tsucc)\right)
        \tag{splitting the sum}\\
    &=
    \left(\sum_{\tsucc \in \post(\state,\action)} \widehat\trans(\state,\action,\tsucc) \cdot \ub(\tsucc)\right)
    +
    1
    -
    \left(\sum_{\tsucc \in \post(\state,\action)} \widehat\trans(\state,\action,\tsucc)\right)
        \tag{$\trans(\state,\action)$ is a probability distribution}\\
    &= 
    \left(\sum_{\tsucc: \#(\state,\action,\tsucc)>0} \widehat\trans(\state,\action,\tsucc) \cdot \ub(\tsucc)\right) + 1 - \left(\sum_{\tsucc: \#(\state,\action,\tsucc)>0} \widehat\trans(\state,\action,\tsucc)\right)
        \tag{$*$}\\
    &= \widehat\ub(\state,\action) 
            \tag{Definition of $\widehat\ub$}\\
    &= \ub'(\state,\action)
            \tag{Line \ref{line:update} in $\UPDATE$}
\end{align*}
Reasoning $(*)$: If for a state $\tsucc$ it holds that $\#(\state,\action,\tsucc)>0$, then $\tsucc \in \post(\state,\action)$, because otherwise the state cannot have been sampled; so no summands are added.
If there are states in $\post(\state,\action)$ that have not been sampled yet, then $\widehat\trans(\state,\action,\tsucc) = 0$ as $\#(\state,\action,\tsucc)=0$; thus these states have no impact on the sum.

So for every action $\val(\state,\action) \leq \ub'(\state,\action)$, so the new estimate is correct.
\qedsquare
\end{proof}

\subsection{Proof of Theorem \ref{thm:anyN}}
\setcounter{theorem}{0}

\begin{theorem}
For any choice of sequence for $\theN$, 
Algorithm \ref{alg:blackBVI} is an anytime algorithm with the following property:
When it is stopped, it returns an interval for $\val(\initstate)$ that is PAC for error tolerance $\delta$ and some $\varepsilon'$, with $0 \leq \varepsilon' \leq 1$.
\end{theorem}
\begin{proof}
%\textcolor{blue}{Note that the statement of the theorem might seem trivial at first glance, as one can easily specify the algorithm that always returns the bounds [0,1]. 
%The theorem also holds for this useless algorithm.
%However, our algorithm computes something useful, and hence it is an achievement that whatever it computes is correct with the given error tolerance.}

We proceed in two steps:
First we prove, that we can continue with the assumption that our partial model is correct, as this assumption is only violated with a probability greater than the given error tolerance $\delta$ (intuitively, this is the 'probably' in the PAC-guarantee).
Then we prove that under this assumption all the computations of our algorithms are correct and return sensible bounds.

\paragraph{Assumption 1:} \emph{Every time the standard BVI (Lines \ref{line:beginBVI}-\ref{line:endBVI}) is executed, it holds that during the whole computation, for all $\state \in \Vis, \action \in \Av(\state), \tsucc \in \post(\state,\action)$ we have that $\widehat\trans(\state,\action,\tsucc) \leq \trans(\state,\action,\tsucc)$.}

We now prove that Assumption 1 is only violated with probability smaller than $\delta$, as sketched in Section \ref{sec:full}.
\begin{itemize}
    \item Every standard BVI depends on a fixed set of simulations, and hence on fixed values of the counters, so during the computation the probability estimates do not change. 
    \item Every standard BVI uses its own error tolerance $\deltaIter$. As argued in Section \ref{sec:full}, the sum of all these error tolerances is $\delta$.
    \item This error tolerance $\deltaIter$ is distributed over all transitions in the model as described in Section \ref{sec:safeUpdate}, where the number of all transitions is over-approximated. 
    Hence the sum of the error probability in the whole partial model is bounded as follows:
    \begin{align*}
        \sum_{\state \in \Vis, \action \in \Av(\state), \tsucc \in \post(\state,\action)} \transdelta &=
        \abs{\set{(\state,\action,\tsucc) \mid \state \in \Vis, \action \in \Av(\state), \tsucc \in \post(\state,\action)}}  ~ \cdot ~ \transdelta\\
        &\leq \abs{\set{\action \mid \state \in \Vis \wedge \action \in \Av(\state)}} \cdot \postmax ~ \cdot ~ \transdelta \tag{using over-approximation of the number of transitions}\\
        &=  \deltaIter \tag{by definition of $\transdelta$ in Line \ref{line:beginBVI}.}
    \end{align*}
\end{itemize}
So in total, the sum of the error probability in each partial model is bounded by $\deltaIter$, and the sum of all $\deltaIter$ is bounded by $\delta$, so the overall error is bounded by $\delta$.
Thus, we continue with Assumption 1.

It remains to prove that every time our algorithms modify the bounds, these modifications are correct.
We do this by an induction as follows:
For every standard BVI phase, the bounds are initialized conservatively in Line \ref{line:reinit}, so they are certainly correct in the beginning, i.e. $\lb(\state)\leq\val(\state)\leq\ub(\state)$.

Assuming that $\lb$ and $\ub$ are correct, we now prove
for each of our algorithms that modify $\lb$ and $\ub$, that after the modification they still are correct.

\begin{itemize}
    \item[$\UPDATE$:] Given the assumption that $\lb$ and $\ub$ are correct and Assumption 1, Lemma 1 proves our goal.
    \item[$\DEFLATE$:] 
    Note that the case of a target being inside the deflated state set $T$ is handled correctly, since in that case all upper bounds remain at 1 and hence are correct over-approximations.
    We continue with the assumption that $\targetset \cap T = \emptyset$.
    We now make a case distinction on whether we know about the $\exit(T)$ of Maximizer in our partial model.
    \begin{itemize}
        \item[We know $\exit$ and it exists:] This means there exists some state-action pair $(\state,\action)$ with $\state \in \states<\Box>$ such that $(\state,\action) \leaves T$,
        and $\widehat\trans(\state,\action,\tsucc)>0$ for some $\tsucc \notin T$.
        Note that by induction hypothesis $\ub(\tsucc)$ is correct.
        Then the value of all states in $T$ is set to $\widehat\ub(\state,\action)$.
        By a similar argument as in the proof of Lemma \ref{lem:updateCorr} we get that 
        $\widehat\ub(\state,\action) \geq \ub(\state,\action)$.
        By \cite[Lemma 2]{KKKW18} we know that $\ub(\state,\action) \geq \val(u)$ for all $u \in T$.
        Thus, after deflating the upper bound still is correct.
        
        \item[There is no $\exit$: ] In this case, Maximizer cannot leave $T$. Note that there is no target in $T$. Hence the value of all states is 0, and deflate sets the upper bound to 0 correctly.
        
        \item[There is a $\exit$, but we do not know it in our partial model:]
        This case violates Assumption 1.
        Before deflate is called on a state set $T$, $\FIND$ ensures that $T$ is an EC $\deltasure\text{ly}$. 
        So if there exists an exiting state-action pair $(\state,\action)$, we have sampled it the number of steps required in Algorithm \ref{alg:sureEC}.
        We arrive at a contradiction by showing that the sum of all the estimated probabilities of the staying actions is larger than the actual sum of all staying probabilities, which violates Assumption 1. 
        This is done by the following chain of equations:
        \begin{align*}
            \sum_{u \in T} \trans(\state,\action,u) &\leq 1 - p_{\min} \tag{since there is at least one $\tsucc \in \post(\state,\action) \setminus T$}\\
            &< 1 - (\abs{\post(\state,\action)} \cdot c) \tag{$c$ is the confidence width, argument for $p_{\min} >\abs{\post(\state,\action)} \cdot c$ below}\\
            &= 1 - \sum_{u \in \post(\state,\action)} c \tag{rewrite}\\
            &< 1 - \sum_{u \in T \cap \post(\state,\action)} c \tag{$(\state,\action)$ is exiting.}\\
            &= \sum_{u \in T \cap \post(\state,\action)} \frac{\#(\state,\action,u)}{\#(\state,\action)} - \sum_{u \in T \cap \post(\state,\action)} c \tag{only sampled successors in $T$, hence left sum adds up to 1}\\
            &= \sum_{u \in T \cap \post(\state,\action)} \frac{\#(\state,\action,u)}{\#(\state,\action)} - c \tag{pulling sums together}\\
            &\leq \sum_{u \in T \cap \post(\state,\action)} \widehat\trans (\state,\action,u) \tag{definition of $\widehat\trans$}\\
            &= \sum_{u \in T} \widehat\trans (\state,\action,u) \tag{Summand for states $\notin \post(\state,\action)$ is 0.}
        \end{align*}
        
        It remains to show $p_{\min} > \abs{\post(\state,\action)} \cdot c$.
        We use $M \eqdef\abs{\post(\state,\action)}$ to improve readability.
        Note that $c \eqdef \sqrt{\frac{\ln(\transdelta)}{-2 \cdot \#(\state,\action)}}$,
        and since $\#(\state,\action) \geq \frac{\ln(\transdelta)}{\ln(1-p_{\min})}$ by the fact that we have a $\deltasure$ EC, it holds that
        $c \leq \sqrt{\frac{\ln(1-p_{\min})}{-2 }}$.
        \begin{align*}
            &p_{\min} > M \cdot \sqrt{\frac{\ln(1-p_{\min})}{-2 }} \\
            \iff& (p_{\min})^2 + \frac{M^2}{2} \cdot \ln(1-p_{\min}) > 0 \tag{squaring, adding right side}
        \end{align*}
        Let us call the left side of this inequation $f(p_{\min})$.
        Note that $f(0)=0$ and for any $M \in \mathbb{N}$ and $p_{\min} \in [0,1): f(p_{\min}) > 0$.
        Thus, since $0 < p_{\min} < 1$, $f(p_{\min}) > 0$ and we proved the claim.
    \end{itemize}
\end{itemize}

\qedsquare
\end{proof}

\subsection{Proof of Theorem 2}\label{app:t2}

\begin{theorem}
There exists a choice of $\theN$, such that 
Algorithm \ref{alg:blackBVI} is PAC for any input parameters $\varepsilon, \delta$, i.e. it terminates almost surely 
and returns an interval for $\val(\initstate)$ of width smaller than $\varepsilon$ 
that is correct with probability at least $1-\delta$.
\end{theorem}
\begin{proof}
Note that in this proof we will provide a lower bound for $\theN$ that is most probably astronomically high and practically infeasible.
This is because the theorem only claims that there exists a choice to get the PAC guarantee. 
For practical purposes, instead of fixing an $\varepsilon$ a priori, it is more reasonable to use the property of Theorem 1, i.e. that typically a good $\varepsilon$ is achieved quickly.

Further, we need to modify the definition of the confidence width $c$ slightly (making our probability estimates more conservative) in order to obtain the guarantee\footnote{We use the two-sided version of the Hoeffding bound instead of the one-sided version. This change entails that the computation of $\widehat\trans$ is different, since $c$ is a larger number. In a future version of this paper, we will adjust the whole paper to use the more conservative probability estimates for which we can prove convergence. However, note that using the one-sided variant has the advantage of being more practical while still being correct.
Theorem 2 actually is only proven for for the variation of Algorithm \ref{alg:blackBVI} that uses the different $\widehat\trans$. For the Algorithm as it is stated in the main body, it is unclear whether it converges.}.

\paragraph{Assumption 2:}
\emph{Every time the standard BVI (Lines 11-16) is executed, it holds that for all $s \in \Vis, a \in \Av(s), t \in \post(s,a)$ we have $\widehat\trans(s,a,t) \geq \trans(s,a,t) - 2c$.}

The difference to Assumption 1 is that Assumption 2 requires a maximum difference between $\trans$ and $\widehat\trans$, namely $2c$.
To achieve this, we use the two sided variant of the Hoeffding bound, i.e.

\begin{align*}
    \Prob\left(\abs{\frac{\#(s,a,t)}{\#(s,a)} - \trans(s,a,t)} \geq c\right) &\leq 2 \cdot e^{-2c^2 \#(s,a)} \tag{\cite[Theorem 1]{hoeffding}}\\
    &\leq \transdelta \tag{Required bound on error probability}
\end{align*}

Solving the last two lines for the confidence width $c$ yields $c \geq \sqrt{\frac{\ln(\transdelta/2)}{-2 \#(s,a)}}$.
Note the difference to the confidence width that is used for the more practical algorithm: the numerator is $\ln(\transdelta/2)$. 

Using the same argument as in the proof of Assumption 1, we get that with probability at least 1-$\delta$, we have 
$\abs{\frac{\#(s,a,t)}{\#(s,a)} - \trans(s,a,t)} < c$, i.e. the difference between the empirical average and the true transition probability is less than $c$.
Thus, we get that $\widehat\trans = \frac{\#(s,a,t)}{\#(s,a)} - c$ is at least $\trans(s,a,t) - 2c$ and at most $\trans(s,a,t)$ (we do not use this upper bound for the proof of convergence, but it shows that Assumption 2 implies Assumption 1).

\medskip

Based on Assumption 2, below we prove that the error of the lower bound in the initial state $\val(\initstate) - \lb(\initstate)$ is at most $\frac \varepsilon 2$. 
The proof for the upper bound is analogous, and we arrive at the conclusion that the width of the returned confidence interval is less than $\varepsilon$ and thus prove the theorem.

It remains to show:
\[
\val(\initstate) - \lb(\initstate) \leq \frac \varepsilon 2 \]

To argue about the error aggregating on a path, we need the game to be acyclic. For this we use the well-known trick of unrolling the game with a step counter.
    
    Let $\G$ be the original game. Then $\G<r>$ is the \emph{unrolled} game for $r$ steps. It is obtained by taking the product of the state space and a step counter, i.e. the new state space is $\states \times [r]$, where $[r]$ is the set of natural numbers from 0 to $r$.
    The action space does not change; ownership of a state and the available actions only depend on the first component (the original state) and not the counter. The transition function works as the original one on the first component, and always increases the counter by one. Thus, the game is acyclic, because there is no transition back to a smaller step counter. 
    Finally, when the counter should be increased beyond $r$, the transition instead goes to a sink state. 

Let $\val(\G)$ denote the value from the initial state of a game. For an unrolled game, the reachability objective is achieved whenever a state with a target state in the first component is reached.

We have $\val(\G<r>) \leq \val(\G)$ and  $\lim_{r \to \infty} \val(\G<r>) = \val(\G)$. 
This is because the value of the unrolled game is the value that can be achieved within $r$ steps, the $r$-horizon reachability. It cannot be greater than the complete reachability, as for the first $r$ steps it is the same, and for the remaining steps it has a reachability probability of 0. Note that value iteration from below computes the $r$-horizon reachability for increasing $r$, and always stays a lower bound. Since value iteration converges in the limit~\cite[Chapter 3.2]{condonAlgo}, we also know that letting $r$ run to infinity, the value of the unrolled game approaches the value of the original game. 

    In the following, let $r$ be a number such that 
    $\val(\G) - \val(\G<r>) \leq \frac \varepsilon 4$, i.e. the error introduced by unrolling the game is at most half of the error that we allow for the lower bound; the other half is used for the error introduced by the probability approximation.

Let $\widehat\G$ be the modification of $\G$ when using $\widehat\trans$ as transition probabilities and redirecting the remaining transition probability to sink states. $\val(\widehat\G)$ is the limit of the partial BVI. In other words, it is the best lower bound that we can compute given a certain fixed $\widehat\trans$.
    Let the number of iterations of partial BVI be at least $r$. We can assume this, as the number of iterations of partial BVI is increased in every iteration of the main loop.
    Then we have that the lower bound that Algorithm 7 computes $\lb(\initstate)$ is at least $\val(\widehat{\G}_r)$.

If we prove $\val(\G<r>) - \val(\widehat{\G}_r) \leq \frac \varepsilon 4$, then we can prove our goal as follows: 
    \begin{align*}
    \val(\initstate) - \lb(\initstate) &\leq
    \val(\G) - \val(\widehat{\G}_r) \tag{By definition of $\val(\G)$ and the previous argument}\\
    &\leq (\val(\G<r>) + \frac \varepsilon 4) - \val(\widehat{\G}_r) \tag{By choice of $r$}\\
    &\leq \frac \varepsilon 4 + \frac \varepsilon 4 \tag{By the fact still to be proven: $\val(\G<r>) - \val(\widehat{\G}_r) \leq \frac \varepsilon 4$}\\
     &= \frac \varepsilon 2
    \end{align*}

It remains to show: $\val(\G<r>) - \val(\widehat{\G}_r) \leq \frac \varepsilon 4$, i.e. in the acyclic unrolled game, the error that is introduced by using $\widehat\trans$ as transition function is less than $\frac \varepsilon 4$ for a suitable choice of the number of samples $\theN$.
In the following, let $\val<r>$ denote the value of a state in $\G<r>$ and $\widehat{\val}_r$ denote the value of a state in $\widehat{\G}_r$. Note that the games $\G<r>$ and $\widehat{\G}_r$ share the state space $\states \times [r]$.

We proceed by induction to show that for every state $s \in \states \times [r]$ we have 
\[
\val<r>(s) - \widehat{\val}_r(s) \leq 2c \cdot (\frac{1}{p_{\min}})^{\abs{\states} \cdot r} \cdot \abs{\states} \cdot r
\]

For a state $s$, let $i$ be the length of the longest path from $s$ to a target or sink state. As the game is acyclic, we have $i \leq \abs{\states} \cdot r$. The claim we inductively prove is  
\[
\val<r>(s) - \widehat{\val}_r(s) \leq 2c \cdot (\frac{1}{p_{\min}})^{i} \cdot i
\]
Since every path has length at most $\abs{\states} \cdot r$, this proves our goal.

\begin{itemize}
        \item[Base case:] Let $i=0$, i.e. the considered state is a target or a sink. Then the difference between $\val<r>(s)$ and $\widehat{\val}_r(s)$ is 0, so we have $ \val<r>(s) - \widehat{\val}_r(s) = 0 = 2c \cdot (\frac{1}{{p_{\min}}})^0 \cdot 0$.
        \item[Induction hypothesis:] For a state $s \in \states \times [r]$ that reaches a target or sink state within at most $i$ steps, we have $\val<r>(s) - \widehat{\val}_r(s) \leq 2c \cdot (\frac{1}{p_{\min}})^{i} \cdot i$
        \item[Induction step:] Consider a Maximizer state $s$ that reaches a target or sink state within at most ($i+1$) steps. 
        Let $a_1 \in \arg \max_{a_1 \in \Av(s)} \val<r>(s,a_1)$ be the action maximizing the original value $\val<r>$ and similarly let $a_2$ maximize $\widehat{\val}_r(s,a_2)$. Then we have:
        \begin{align*}
             \val<r>(s) - \widehat{\val}_r(s) &= 
             \val<r>(s,a_1) - \widehat{\val}_r(s,a_2) \tag{picking the maximizing actions}\\
             &\leq \val<r>(s,a_1) - \widehat{\val}_r(s,a_1) \tag{since $\widehat{\val}_r(s,a_1) \leq \widehat{\val}_r(s,a_2)$}\\
             &=\sum_{s'} (\trans(s, a_1, s')\val<r>(s') - \widehat\trans(s, a_1, s')\widehat{\val}_r(s')) \tag{using definitions of $\val<r>$ and $\widehat{\val}_r$}\\
             &\leq 
              \frac{1}{p_{\min}} (\trans(s, a_1, t)\val<r>(t) - \widehat\trans(s, a_1, t)\widehat{\val}_r(s')).
        \end{align*}

        The last step holds, as there are at most $\frac{1}{p_{\min}}$ successors of state $s$, where $t$ is that state in $\post(s,a_1)$ that maximizes $(\trans(s, a_1, t)\val<r>(t) - \widehat\trans(s, a_1, t)\widehat{\val}_r(t))$.
        
        If we show
        \[\trans(s, a_1, t)\val<r>(t) - \widehat\trans(s, a_1, t)\widehat{\val}_r(t) \leq 2c~\cdot~(\frac{1}{p_{\min}})^i \cdot (i+1), \]
        we can conclude that 
        \begin{align*}
            \val<r>(s) - \widehat{\val}_r(s) &\leq
            \frac{1}{p_{\min}} (\trans(s, a_1, t)\val<r>(t) - \widehat\trans(s, a_1, t)\widehat{\val}_r(s'))\\
            &\leq \frac{1}{p_{\min}} \cdot 2c~\cdot~(\frac{1}{p_{\min}})^i \cdot (i+1) \\
            &= 2c~\cdot~(\frac{1}{p_{\min}})^{i+1} \cdot (i+1)
        \end{align*}
        
        Note that $t$ is a successor of $s$, and since the game is acyclic, it reaches a target or sink state within at most $i$ steps. Thus we can apply the induction hypothesis on $\val<r>(s') - \widehat{\val}_r(s')$ and use the following chain of inequations to prove our goal:
        \begin{align*}
        &\trans(s, a_1, t)\val<r>(t) - \widehat\trans(s, a_1, t)\widehat{\val}_r(t)\\
        \leq~&\trans(s, a_1, t)\val<r>(t) - (\trans(s, a_1, t) - 2c) \widehat{\val}_r(t)\tag{since $\widehat\trans(s,a_1,t) \geq \trans(s,a_1,t) - 2c$ by Assumption 2}\\
        =~&\trans(s, a_1, t) \cdot (\val<r>(t) - \widehat{\val}_r(t)) + 2c \cdot \widehat{\val}_r(t) 
        \tag{using distributivity twice}\\
        \leq~& \trans(s, a_1, t) \cdot (2c \cdot (\frac{1}{{p_{\min}}})^i \cdot i )+ 2c \cdot \widehat{\val}_r(t) \tag{by induction hypothesis}\\
        =~& 2c \cdot (\trans(s, a_1, t) \cdot (\frac{1}{{p_{\min}}})^i \cdot i + \widehat{\val}_r(t) ) \tag{distributivity}\\
        \leq~& 2c \cdot ( (\frac{1}{{p_{\min}}})^i \cdot i + 1) \tag{since $\trans$ and $\widehat{\val}_r$ are at most 1}\\
        =~& 2c \cdot (\frac{1}{{p_{\min}}})^i \cdot (i + p_{\min}^i) \tag{distributivity}\\
        \leq~& 2c \cdot (\frac{1}{{p_{\min}}})^i \cdot (i + 1) \tag{since $p_{\min} \leq 1$}
        \end{align*}
    \end{itemize}
    
    Dually, the same can be done for a state of Minimizer.
    
    So the induction is complete, and we arrive at the conclusion that the error in the initial state can be bounded.
    Thus we get the following expression for the error between the games 
    \[
    \val(\G<r>) - \val(\widehat{\G}_r) \leq 2c \cdot (\frac{1}{p_{\min}})^{\abs{\states} \cdot r} \cdot \abs{\states} \cdot r.
    \]

It remains to show that we can choose an $\theN$ such that $2c \cdot (\frac{1}{p_{\min}})^{\abs{\states} \cdot r} \cdot \abs{\states} \cdot r \leq \frac \varepsilon 4$.

Reordering the inequation yields:
\[ 
c \leq \frac \varepsilon 8 \cdot p_{\min}^{\abs{\states} \cdot r} \cdot \frac{1}{\abs{\states} \cdot r}
\]

The confidence width is given by $c \geq \sqrt{\frac{\ln(\transdelta/2)}{-2 \#(s,a)}}$ (recall that we needed the more conservative confidence width for Assumption 2).

Therefore, denoting $\#(s,a)$ by $n$, we have
\begin{align*}
    \sqrt{\frac{\log(\transdelta/2)}{-2n}} &\leq \frac \varepsilon 8 \cdot p_{\min}^{\abs{\states} \cdot r} \cdot \frac{1}{\abs{\states} \cdot r}\\
    \implies n &\geq \frac{-32\log(\transdelta/2) \cdot \abs{\states}^2 \cdot r^2}{\varepsilon^2 \cdot p_{\min}^{2\abs{\states}\cdot r}}
\end{align*}

In other words if every relevant transition\footnote{
    Note that the guiding heuristic picks the correct actions to sample the relevant parts of the state space, since by Theorem \ref{thm:anyN} the bounds are correct and by \cite[Theorem 3]{BCC+14} simulation based asynchronous value iteration converges.}
    in the partial model is chosen at least $n$ times, then we would achieve the desired confidence width.
    Then the error of $\val(\widehat{G}_r)$ is bounded by $\frac \varepsilon 4$, the overall error of the lower bound by $\frac \varepsilon 2$, and by repeating the argument for the upper bound, the overall error by $\varepsilon$.
    
    It remains to prove that we choose an $\theN$ such that every relevant transition is played at least $n$ times ($\theN$ is a constant and independent of $k$).

    Let $p = p_{\min}^{\abs{\states}}$ be the minimum probability with which some transition can be played. 
    We want to compute the number of samples $\theN$ required so that every transition is triggered at least $n$ times.
    Let $X$ be a random variable counting the number of times the least likely transition has been triggered. 
    $X$ is distributed according to a Binomial distribution with parameters $\theN$ and $p$. 
    The cumulative distribution function $F(X \leq n) = \sum_{i=0}^n \binom{\theN}{i} p^i (1-p)^{n-i}$ gives the probability that the least likely transition is triggered at least $k$ times when $N$ samples are taken.
    If we choose an $\theN$ such that $F(X \geq n) \geq 1/2$, then it means that every relevant transition will be triggered at least $n$ times with a probability of $\frac{1}{2}$, which means that the partial BVI will succeed with a probability of $\frac{1}{2}$. 
    If the partial BVI fails, then we repeatedly try it until it succeeds.
    Note that hence we succeed almost surely, since the probability of succeeding in any BVI phase is $\sum_{i=1}^\infty \frac 1 {2^i}$.
    An $\theN$ for which $F(X \geq n) \geq \frac{1}{2}$ exists.

\end{proof}

}{}

\end{document}